\documentclass[11pt]{amsart}
\usepackage{amsmath, amssymb, amsthm, url}

\usepackage{graphicx}
\allowdisplaybreaks

\usepackage{graphicx}

\usepackage{amsmath}
\usepackage[mathscr]{eucal}
\usepackage{amscd}
\usepackage{amsfonts,enumerate,epsfig,fancyhdr,subfigure}
\usepackage{mathrsfs}
\usepackage{multirow}
\usepackage{multicol}
\usepackage{pdflscape} 
\usepackage{longtable}
\usepackage[utf8]{inputenc}
\usepackage{color}
\usepackage{mathtools}

\usepackage{xcolor,colortbl}   

\theoremstyle{plain}
\newtheorem{theorem}{Theorem}[section]
\newtheorem{proposition}[theorem]{Proposition}
\newtheorem{lemma}[theorem]{Lemma}
\newtheorem{corollary}[theorem]{Corollary}

\theoremstyle{definition}
\newtheorem{definition}[theorem]{Definition}
\newtheorem{example}[theorem]{Example}

\newtheorem{remark}[theorem]{Remark}

\def\x{{\mathbf x}}
\def\y{{\mathbf y}}

\def\0{{\mathbf 0}}
\def\1{{\mathbf 1}}

\def\g{{\mathbf g}}

\def\s{{\sigma}}

\def\g{{\gamma}}

\def\M{{\mathcal M}}
\def\F{{\mathbb F}}

\def\D{{\mathbb D}}
\def\0{{\mathbf 0}}
\def\1{{\mathbf 1}}

\def\CC{{\mathcal C}}

\def\Aut{\operatorname{Aut}}

\def\x{{\mathbf x}}
\def\y{{\mathbf y}}
\def\c{{\mathbf c}}

\def\0{{\mathbf 0}}
\def\1{{\mathbf 1}}

\def\g{{\mathbf g}}

\def\s{{\sigma}}

\def\g{{\gamma}}

\def\M{{\mathcal M}}
\def\F{{\mathbb F}}

\def\D{{\mathbb D}}
\def\0{{\mathbf 0}}
\def\1{{\mathbf 1}}

\def\CC{{\mathcal C}}

\def\Aut{\operatorname{Aut}}

\newenvironment{smatrix}{\left(\begin{smallmatrix}}{\end{smallmatrix}\right)}

\begin{document}

	\title[ Self-dual codes over $GF(q)$ with symmetric generator matrices ]
	{Self-dual codes over $GF(q)$ with symmetric generator matrices}

	\author{Jon-Lark Kim${}^{*}$ \and  Whan-Hyuk Choi${}^{**\dag}$  }

	\address{Department of Mathematics, Sogang University, Seoul 04107, Republic of Korea, \ Email: whchoi@kangwon.ac.kr}
	
	\address{Department of Mathematics, Sogang University, Seoul 04107, Republic of Korea, \ Email: jlkim@sogang.ac.kr }
	
	\date{}
	\subjclass[2000]{Primary: 94B05, Secondary: 11T71}
	
	\keywords{symmetric self-dual code, optimal codes, self-dual codes, symmetric generator matrix, quadratic residue codes}

	\thanks{\tiny   ${}^{\dag}$Corresponding author. \\
		The author${}^{*}$ is supported by the
		National Research Foundation of Korea (NRF) grant funded by the
		Korea government (NRF-2019R1A2C1088676).
		The author${}^{**}$ is supported by the
		National Research Foundation of Korea (NRF) grant funded by the
		Korea government (NRF-2019R1I1A1A01057755).}


	\maketitle

\begin{abstract}
	We introduce a consistent and efficient method to construct self-dual codes over $GF(q)$ with symmetric generator matrices from a self-dual code over $GF(q)$ of smaller length where $q \equiv 1 \pmod 4$. Using this method, which is called a `symmetric building-up' construction, we improve the bounds of best-known minimum weights of self-dual codes of lengths up to 40, which have not significantly improved for almost two decades. We focus on a class of self-dual codes, which includes double circulant codes. We obtain 2967 new self-dual codes over $GF(13)$ and $GF(17)$ up to equivalence. Besides, we compute the minimum weights of quadratic residue codes that were not known before. These are: a [20,10,10] QR self-dual code over $GF(23)$, [24,12,12] QR self-dual codes over $GF(29)$ and $GF(41)$, and a [32,16,14] QR self-dual codes over $GF(19)$. They have the highest minimum weights so far.

\end{abstract}
	
	
	\section{Introduction}

	%
	
	The theory of error-correcting code, which was born with the invention of computers, has been an interesting topic of mathematics as well as industry, such as satellites, CD players, and cellular phones. Recently, with the advent of machine learning and artificial intelligence, there have been some studies on the relationship between error-correcting codes and these fields \cite{BeEry2019,Huang2019,Mills2019,Nachmani2018}.
	Especially, self-dual codes have been an important class of linear codes for both practical and theoretical reasons and received an enormous research effort from the beginning of coding theory. Many of the best-known codes are actually self-dual codes. It is well-known that self-dual codes are asymptotically good \cite{MacWilliam1972}. Moreover, self-dual codes also have close connections to other mathematical structures such as designs, lattices,  graph theory, and modular forms \cite{Bannai1999,Conway1999}. Recently, self-dual codes have applications in quantum information theory \cite[Chap. 13]{Nebe2006}.

	On the other hand, coding theorists are interested in finding an {\it optimal} code, which has the best capability to correct as many errors as possible with a given length. The {\it minimum distance} of code is the parameter determining the error-correction capability of a code. In particular, {\it extremal} self-dual codes and {\it maximal distance separable (MDS)} self-dual codes are optimal codes that meet some upper bounds of minimum distance. We refer to \cite{arasu2001,betsumiya2003,Gaborit2003,Gaborit2002,Georgiou,grassl2008,Guenda2012,huffman2005,sok2019,sok2020,Yan2019}. 
	
	As a summary, we present all of the up-to-date results concerning minimum weight bounds and the existence of optimal self-dual codes in Tables \ref{former_result_A}, \ref{former_result_B}, and \ref{former_result_C}. In these tables, the bounds of the highest minimum weight are listed. The superscript `$e$'  indicates the extremal code, and `$*$' indicates the MDS code. The superscript `$o$' indicates there are no extremal or MDS codes, but the minimum distance is proved to be optimal with given parameters. If the bound is not determined yet, we inscribe `$?$' and if there exists no self-dual code, we inscribe `-'. In Tables \ref{former_result_A}, we list best-known Lee distances($d_L$) and Hamming distances($d_H$) of euclidean self-dual codes over $GF(4)$(denoted by $4^{eucl}$) and best-known Hamming distances of hermitian self-dual codes over $GF(4)$(denoted by $4^{herm}$).
	
	Gleason-Pierce-Ward theorem states that self-dual codes over $GF(q)$ have weights divisible by $\delta >1$ only if $q=2,3,4$. This motivates many researchers to study self-dual codes over small fields. Table \ref{former_result_A} gives the updated status of the highest minimum weights of such self-dual codes.
	However, these tables also tell that there remain many unknown bounds. Most cases of length $\le$ 24 are completely known. However, when $5 \le q \le 20$, most highest minimum weights of self-dual codes over $GF(q)$ are not known if length $\ge 24$, as we can see in Table \ref{former_result_B} and Table \ref{former_result_C}. However, in general, many self-dual codes over larger finite fields have better minimum weights than those of self-dual codes over smaller fields. This is the main motivation of this paper. 
	
	We try to improve the bounds of minimum weights by constructing self-dual codes of long length as many as possible. To this end, we investigate the consistent and efficient method to construct self-dual codes. Consequently, we find a construction method of self-dual code over $GF(q)$ having a symmetric generator matrix where $q \equiv 1 \pmod 4$. This method can be regarded as a special case of the well-known `building-up' construction method \cite{Kim2004}. However, the method in this paper has significant differences: we improve the efficiency to find the best self-dual code from a self-dual code of given length and we also focus our concern on one subclass of self-dual codes which have a certain automorphism in their automorphism group. Using this construction method, we obtain 2967 new self-dual codes over $GF(13)$ and $GF(17)$ and improve the lower bounds of best self-dual codes of length up to 40 (Table \ref{our_results1} and \ref{numbers}). We also want to point out that our new construction method includes well-known pure double circulant and bordered double circulant construction; for example, optimal and MDS self-dual codes obtained in \cite{betsumiya2003} and \cite{Gulliver2008} can be obtained equivalently by using our method.
	
	In addition, we construct four new self-dual codes from quadratic residue codes which improve the unknown bound: a [20,10,10] code over $GF(23)$, [24,12,12] codes over $GF(29)$ and $GF(41)$, and [32,16,14] codes over $GF(19)$. We also point out that the quadratic residue code over $GF(13)$ of length 18, which has been reported previously as the optimal self-dual code(\cite{betsumiya2003}), is {\it not} actually a self-dual code. However, since we obtain [18,9,8] self-dual codes over $GF(13)$, the bound of the highest minimum distance of self-dual code over $GF(13)$ of length 18 is turned to 8-9. Our new results are written in bold in Tables \ref{former_result_B}, \ref{former_result_C} and \ref{our_results1}. In particular, the highest minimum distances of our results in Table \ref{our_results1} are all of the self-dual codes having symmetric generator matrices. The number of inequivalent codes we obtain is given in Table \ref{numbers}.

	\begin{table}
		\begin{center}
			\begin{small}
				\begin{tabular}{|>{\centering\arraybackslash}p{0.7cm}||>{\centering\arraybackslash}p{1.3cm}|>{\centering\arraybackslash}p{1.3cm}|>{\centering\arraybackslash}p{1.3cm}|>{\centering\arraybackslash}p{1.3cm}|>{\centering\arraybackslash}p{1.3cm}|>{\centering\arraybackslash}p{1.3cm}|}
					\hline
					\multirow{2}{*}{$n\backslash q$}
					&\multicolumn{2}{c|}{$2$}     &\multirow{2}{*}{3}    &  \multicolumn{2}{c|}{$4^{eucl}$}    & \multirow{2}{*}{$4^{herm}$}    \\

					\cline{2-3}\cline{5-6}
					
					&   type I        &    type II            &           &     $d_L$             &  $d_H$             &  \\
					\hline
					2  &$2^*$ &-                    &- 				   &$2$			&$2^*$   	   &2       	\\\hline
					4  &$2^o$ &- 					&$3^*$			&$2$		&$3^*$		&2	\\\hline
					6  &$2^o$ &-                    &-				 	  &$4$&$3^o$     &4            	\\\hline
					8  &$2^o$ &$4^e$				&$3^e$		&$4$&$4^e$	&4		\\\hline
					10  &$2^o$ &-                    &- 			     &$4$&$4^e$     &4        	\\\hline
					12  &$4^e$ &- 					&$6^e$			&$6$&$6^o$	&4		\\\hline
					14  &$4^e$ &-                    &-				     &$6$&$6^o$      &6          	\\\hline
					16  &$4^e$ &$4^e$				&$6^e$		&$6$&$6^o$		&6		\\\hline
					18  &$4^e$ &-                    &- 			      &$8$&$6-7$       &8         	\\\hline
					20  &$4^e$ &- 					&$6^e$			&$8$&$8^e$		&8		\\\hline
					22  &$6^e$ &-                    &-				      &$8$&$8^e$       &8         	\\\hline
					24  &$6^e$ &$8^e$				&$9^e$					&$?$&$8-10$		&8	\\\hline
					26  &$6^o$ &-                    &- 			     &$?$&$8-10$    &8,10        	\\\hline
					28  &$6^o$ &- 					&$9^e$			 	&$?$&$9-11$		&10		\\\hline
					30  &$6^o$ &-                    &-				    &$?$&$10-12$    &12       	\\\hline
					32  &$8^e$ &$8^e$				&$9^e$				&$?$&$11-12$	&10,12		\\\hline
					34  &$6^o$ &-                    &- 			    &$12$&$10-12$  &10,12            	\\\hline
					36  &$8^e$ &- 					&$12^e$			 	&$?$&$11-14$	&12,14		\\\hline
					38  &$8^e$ &-                    &-				    &$?$&$11-15$  &12,14        	\\\hline
					40  &$8^e$ &$8^e$				&$12^e$				&$?$&$12-16$	&12,14		\\\hline
					
				\end{tabular}
			\end{small}
			\caption{The best-known minimum weights of self-dual codes of length $n$ over $GF(q)$ where $n \le 40$ and $2 \le q \le 4$ \cite{Gaborit2003,grassl2009,haradaWeb,huffman2005}. }
			\label{former_result_A}
		\end{center}
	\end{table}

	\begin{table}
		\begin{center}
			\begin{small}
				\begin{tabular}{|>{\centering\arraybackslash}p{0.7cm}||>{\centering\arraybackslash}p{1.3cm}|>{\centering\arraybackslash}p{1.3cm}|>{\centering\arraybackslash}p{1.3cm}|>{\centering\arraybackslash}p{1.3cm}|>{\centering\arraybackslash}p{1.3cm}|>{\centering\arraybackslash}p{1.3cm}|>{\centering\arraybackslash}p{1.3cm}|>{\centering\arraybackslash}p{1.3cm}|}
					\hline
					$n\backslash q$ &5     &7         &9          &11                     &13              & 17     & 19     \\\hline
					\hline
					2  &$2^*$ &-                      &$2^*$ 	&-                   &$2^*$          &$2^*$  & -      \\\hline
					4  &$2^o$ &$3^*$					&$3^*$	&$3^*$				 &$3^*$        		&$3^*$  &$3^*$  \\\hline
					6  &$4^*$ &-                      &$4^*$	&-                   &$4^*$          &$4^*$  & -      \\\hline
					8  &$4^o$ &$5^*$					&$5^*$	&$5^*$				 &$5^*$        		&$5^*$  &$5^*$  \\\hline
					10 &$4^o$ &-                      &$6^*$	&-                   &$6^*$          &$6^*$  & -      \\\hline
					12 &$6^o$ &$6^o$					&$6^o$	&$7^*$				 &$6^o$           	&$7^*$  &$7^*$  \\\hline
					14 &$6^o$ &-                      &$6-7$	&-                   &$8^*$             &$7-8$  & -      \\\hline
					16 &$7^o$ &$7-8$					&$8^o$	&$8^o$				 & $8^o$       		 &$8-9$  				&$8-9$  \\\hline
					18 &$7^o$  &-                      &$8-9$	& -                  &$\mathbf{8-9}$     &$10^*$ 				& -      \\\hline
					20 &$8^o$  &$9-10$					&$10^o$	&$10^o$				 &$10^o$        	 &$10^o$ 				&$11^*$ \\\hline	
					22 &$8^o$  &-                      &$9-11$	& -                  &$10-11$        	 &$10-11$				& -      \\\hline
					24 &$9-10$&$9-11$					&$10-11$&$9-12$			     &$10-12$   		 &$10-12$				&$10-12$     \\\hline
					26 &$9-10$&-                      &$10-12$	& -                  & $\mathbf{10-13}$  &$\mathbf{10-13}$         & -      \\\hline
					28 &$10-11$&$11-13$					&$12-13$&$10-14$ 			 & $\mathbf{11-14}$  &$\mathbf{11-14}$      &$11-14$     \\\hline
					30 &$10-12$&-                     &$12-14$	& -                  & $\mathbf{11-15}$  &$\mathbf{12-15}$        & -      \\\hline
					32 &$11-13$&$13-14$					&$12-15$&?                   & $\mathbf{12-16}$  &$\mathbf{12-16}$        &$\mathbf{14-16}$       \\\hline
					34 &$11-14$&-                      &$12-16$	&-                   & $\mathbf{12-17}$  &$\mathbf{13-17}$        &-       \\\hline
					36 &$12-15$&$13-17$					&$13-17$&?                   & $\mathbf{13-18}$  &$\mathbf{13-18}$        &?       \\\hline
					38 &$12-16$&-                      &$14-18$&-                    & $\mathbf{13-19}$  &$\mathbf{14-19}$        &-       \\\hline
					40 &$13-17$&$13-18$					&$14-18$&?                   & $\mathbf{14-20}$  &$\mathbf{14-20}$        &?       \\\hline
				\end{tabular}
			\end{small}
			\caption{The best-known minimum weights of self-dual codes of length $n$ over $GF(q)$ where $n \le 40$ and $5 \le q \le 19$ \cite{betsumiya2003,deboer1996,Gaborit2003,grassl2008,grassl2009,Han2008,harada2003,leon1982,Shi2018}. New results from this article written in bold.}
			\label{former_result_B}
		\end{center}
	\end{table}

	\begin{table}
		\begin{center}
			\begin{small}
				\begin{tabular}{|>{\centering\arraybackslash}p{0.7cm}||>{\centering\arraybackslash}p{1.3cm}|>{\centering\arraybackslash}p{1.3cm}|>{\centering\arraybackslash}p{1.3cm}|>{\centering\arraybackslash}p{1.3cm}|>{\centering\arraybackslash}p{1.3cm}|>{\centering\arraybackslash}p{1.3cm}|>{\centering\arraybackslash}p{1.3cm}|}
					\hline
					$n\backslash q$&23                 &25                     &27                 &29               &31                     &37          & 41        \\\hline
					\hline
					2 &-                     &$2^*$   				&-                  &$2^*$ 				 	&-                      &$2^*$       &$2^*$         \\\hline
					4&$3^*$					&$3^*$   				&$3^*$ 				&$3^*$   				&$3^*$  				 &$3^*$       &$3^*$         \\\hline
					6  &-                       &$3^*$   				&-                  &$4^*$  				&-                      &$4^*$       &$4^*$         \\\hline
					8  &$5^*$ 					&$5^*$   				&$5^*$ 				&$5^*$  				&$5^*$  				 &$5^*$       &$5^*$         \\\hline
					10&-                     &$6^*$   			 	&-                  &$6^*$  				&-                      &$6^*$       &$6^*$         \\\hline
					12 &$7^*$ 					&$7^*$   				&$7^*$ 				&$7^*$   				&$7^*$ 				 	 &$7^*$        &$7^*$         \\\hline
					14 &-                      &$8^*$   				&-                  &$8^*$   				&-                       &$8^*$       &$8^*$         \\\hline
					16&$9^*$ 					&$9^*$   				&$9^*$ 				&$9^*$  				&$9^*$ 			  	 &$9^*$  	&$9^*$   	\\\hline
					18 	&-                      &$10^*$  				&-	                &$10^*$  				&-                    &$10^*$   &$10^*$   \\\hline
					20 	&$\mathbf{10-11}$  		&$11^*$ 				&?                  &$10-11$   				&$11^*$   				 &?          &$11^*$        \\\hline	
					22 &-                     &?       				&-                  &?                      &-                      &?         &$12^*$        \\\hline
					24 	&$13^*$  				&$12-13$ 				&?             		&$\mathbf{12-13}$       &$13^*$    				&?        &$\mathbf{12-13}$     \\\hline
					26  &-                      &$14^*$ 				&-                  &?                      &-                      &     $14^*$        &?                          \\\hline
					28 &$11-14$ 				&?              		&$15^*$ 			&$14-15$ 				&?          				&?        &?             \\\hline
					30&-                     &?                      &-                  &$16^*$     			&-                      &?           &?                          \\\hline
					32 &?                 &?                      &?                  &?                      &$17^*$				 &?          &$17^*$          \\\hline
					34 &-                      &?                      &-                  &?                      &-                      &?           &?                          \\\hline
					36 &?                      &?                      &?                  &?                      & ?                     &$18-19$ 		&?                          \\\hline
					38 	&-                      &?                      &-                  &?                      &-                      &$20^*$     &?                          \\\hline
					40 &?                      &?                      &?                  &?                      & ?                     &?          &$20-21$ \\\hline
				\end{tabular}
			\end{small}
			\caption{The best-known minimum weights of self-dual codes of length $n$ over $GF(q)$ where $n \le 40$ and $23 \le q \le 41$\cite{betsumiya2003,deboer1996,Georgiou,grassl2008,Gulliver2008,gulliver2010,Kim2004,Shi2018,sok2019,sok2020}. New results from this article written in bold. }
			\label{former_result_C}
		\end{center}	
	\end{table}


	
	\begin{table}
		\begin{center}
			\begin{small}
				\begin{tabular}{|c||c|c|c|c|}
					\hline
					\multirow{2}{*}{$n$} &	      \multicolumn{2}{c|}{Over $GF(13)$}   &    \multicolumn{2}{c|}{Over $GF(17)$}     \\
					\cline{2-5}
					&      Our results              &       Prev. best             &    Our results              &       Prev. best    \\
					\hline
					2  &$2$ 	& 2             &2 			&2	          	\\\hline
					4  &$3$ 	& 3 			&3			&3			\\\hline
					6  &$4$ 	& 4             &4			&4 	              	\\\hline
					8  &$5$ 	& 5				&5			&5				\\\hline
					10  &$6$ 	& 6             &6 			&6            	\\\hline
					12  &$6$ 	& 6				&7			&7 				\\\hline
					14  &$8$ 	& 8             &7			&7              	\\\hline
					16  &$8$ 	&8				&8			&8				\\\hline
					18  &$8$	 & 9?           &10 			&10              	\\\hline
					20  &$10$ 	& 10 			&9			&10				\\\hline
					22  &$10$ 	& 10            &10			&10             	\\\hline
					24  &$10$	&10             &10			&10			\\ \hline
					26  &$\mathbf{10}$ &-       &$\mathbf{10}$			&-            	\\\hline
					28  &$\mathbf{11}$ & 10		&$\mathbf{11}$			&10		\\\hline
					30  &$\mathbf{11}$ &-       &$\mathbf{12}$			&-          	\\\hline
					32  &$\mathbf{12}$ &-		&$\mathbf{12}$			&-			\\\hline
					34  &$\mathbf{12}$ &-       &$\mathbf{12}$			&-               	\\\hline
					36  &$\mathbf{13}$ &- 		&$\mathbf{13}$			&-			\\\hline
					38  &$\mathbf{13}$ &-       &$\mathbf{14}$			&-           	\\\hline
					40  &$\mathbf{14}$ & -		&$\mathbf{14}$			&-			\\\hline
					
				\end{tabular}
				\caption{Highest minimum weights of self-dual codes constructed by Theorem \ref{Building up} vs. previously known highest minimum weights. New results are written in bold.}
				\label{our_results1}
			\end{small}
		\end{center}	
	\end{table}

	\begin{table}
		\begin{center}
			\begin{small}
				\begin{tabular}{|c||c|c|c|c|}
					\hline
					\multirow{2}{*}{$n$} &	      \multicolumn{2}{c|}{Over $GF(13)$}   &    \multicolumn{2}{c|}{Over $GF(17)$}     \\
					\cline{2-5}
					& min. wt.         & \# of codes   & min. wt.  & \# of codes   \\
					\hline
					26  &10 &     $\ge 1098	$		&10			&$\ge 352$           	\\\hline
					28  &11 & $\ge 1$		&11			&$\ge 106$	\\\hline
					30  &11 &   $\ge 380$    &12			&$\ge 2$          	\\\hline
					32  &12 &$\ge 164$		&12			&$\ge 2$			\\\hline
					34  &12 &$\ge 710$       &12			&$\ge 2$               	\\\hline
					36  &13 &$\ge 7$	&13			&$\ge 64$			\\\hline
					38  &13 &$\ge 66$       &14			&$\ge 2$          	\\\hline
					40  &14 & $\ge 4$	&14			&$\ge 7$			\\\hline
					
				\end{tabular}
				\caption{Number of inequivalent self-dual codes newly obtained by using construction method of Theorem \ref{Building up} }
				\label{numbers}
			\end{small}
		\end{center}	
	\end{table}

	The paper is organized as follows. Section 2 gives preliminaries and background for self-dual codes over $GF(q)$.
	In Section 3, we present a construction method of {\it symmetric self-dual codes} over $GF(q)$ where $q \equiv 1 \pmod 4$. We show that every symmetric self-dual code of length $2n+2$ is constructed from a symmetric self-dual code of length $2n$ up to equivalence by using this construction method. In Section 4, we present the computational results of the best codes obtained by using our method.
	All computations in this paper were done with the computer algebra system \textsc{Magma} \cite{Magma}.

	\section{Preliminaries}

	Let $n$ be a positive integer and $q$ be a power of a prime. A {\it linear code} $\CC$ of length $n$ and dimension $k$ over a finite field $GF(q)$ is a $k$-dimensional subspace of $GF(q)^n$. An element of $\CC$ is called a {\it codeword}. A {\it generator matrix} of $\CC$ is a matrix whose rows form a basis of $\CC$. 
	For vectors $\x = (x_i )$ and $\y=(y_i)$, we define the inner product $\x \cdot \y = \sum_{i=1}^{n} x_i y_i$. The {\it dual code $\CC^{\perp}$} is defined by $$\CC^{\perp}=\{\x\in GF(q)^n \mid \x \cdot \c=0 \text{ for all $\c \in C$} \}.$$  A linear code $\CC$ is called {\it self-dual} if $\CC = \CC^{\perp}$ and {\it self-orthogonal} if $\CC \subset \CC^{\perp}$. 
	
	The {\it weight} of a codeword $\c$ is the number of non-zero symbols in the codeword and denoted by $wt(\c)$. The {\it Hamming distance} between two codewords $\x$ and $\y$ is defined by $d(\x,\y)=wt(\x-\y)$. The {\it minimum distance} of $\CC$, denoted by $d(\CC)$, is the smallest Hamming distance between distinct codewords in $\CC$. The error-capability of a code is determined by the minimum distance, thus the minimum distance is the most important parameter of a code. For linear codes, the minimum distance equals the minimum weight of the non-zero codewords. It is well-known \cite[chapter 2.4.]{HP3} that a linear code of length $n$ and dimension $k$ satisfy the Singleton bound, $$d(\CC) \le n - k +1.$$ A code that achieves the equality in Singleton bound is called a \textit{maximum distance separable(MDS)} code. A self-dual code of length $2n$ over a field is MDS if the minimum weight equals $n+1$.  
	
	Let $S_n$ be a symmetric group of order $n$ and $\D^n$ be the set of diagonal matrices over $GF(q)$ of order $n$, $$ \D^n=\{diag(\g_i) \mid  \g_i\in GF(q), \g_i^2=1\}.$$ The group of all  \textit{$\gamma$-monomial transformations of length $n$},  $\M^n$  is defined by
	$$\M^n=\{p_{\s} \gamma \mid \gamma \in\D^n, \s \in S_n\}
	$$ where $p_{\sigma}$ is the permutation matrix corresponding $\sigma \in S_n$. 
	We only consider $\gamma$-monomial transformation in this paper since $\gamma$-monomial transformation does preserve the self-duality(see \cite[Thm 1.7.6]{HP3}). Let $\CC\tau =\{\c \tau \mid \c \in \CC\}$ for an element $\tau$ in $\M^{2n}$ and a code $\CC$ of length $2n$. If there exists an element $\mu \in \M^{2n}$ such that $\CC \mu=\CC'$ for two distinct self-dual codes $\CC$ and $\CC'$, then $\CC$ and $\CC'$ are called \textit{equivalent} and denoted by $\CC \simeq \CC'$ . An {\it automorphism} of $\CC$ is an element $\mu \in \M^{2n}$ satisfying $\CC \mu=\CC$. The set of all automorphisms of $\CC$ forms the {\it automorphism group} $\Aut(\CC)$ as a subgroup of $\M^{2n}$.

	Let $A^T$ denote the transpose of a matrix $A$. A self-dual code $\CC$ of length $2n$ over $GF(q)$ is equivalent to a code with a standard generator matrix
	\begin{equation}\label{std-form}
		\left(
		\begin{array}{c|c}
			I_{n}& A
		\end{array}
		\right),
	\end{equation} where $A$ is a $n \times n$ matrix satisfying $AA^T=-I_{n}$. 
	
	%
	%
	%
	%
	%
	%


	\begin{proposition}\label{transpose}
		Let $\CC$ be a self-dual code of length $2n$ over $GF(q)$ with a standard generator matrix
		$G=(    I_n \mid  A ).$ Then $$A^TG =(A^T\mid  -I_n )$$ is also a generator matrix of $\CC$. 
		
		\begin{proof}
			Since $\CC$ is self-dual, $AA^T=-I$ and $A^{-1} =-A^T$. Thus $A^T$ is non-singular. This implies that the rows of matrix $A^TG$ form a basis of the code $\CC$ and $$A^TG =(    A^T I_n \mid  A^TA ) = (    A^T  \mid  -I_n).$$ \end{proof}    
	\end{proposition}

	\begin{corollary}\label{tranpose2}
		Let $G=(    I_n \mid  A )$  and $G'=(    I_n \mid  A^T )$ be generator matrices of self-dual codes $\CC$ and $\CC'$, respectively. Then $\CC$ and $\CC'$ are equivalent.
		\begin{proof}
			By the Proposition \ref{transpose}, it is clear that $G'$ is equal to $G p_{\tau_1} \gamma_1 $ for $\tau_1 = (1,n+1)(2,n+2)\cdots(n,2n) \in S_{2n}$ and $\gamma_1 = diag(-\1_n,\1_n)\in \D^{2n}$ where $\1_n$ denotes all one vector of length $n$. 
		\end{proof}    
	\end{corollary}

	\begin{proposition}\label{M-equivalent}
		Let $G=(    I_n \mid  A )$  and $G'=(    I_n \mid  B )$ be generator matrices of self-dual codes $\CC$ and $\CC'$, respectively. If $A = \mu_1 B \mu_2$ for some $\mu_1, \mu_2 \in \M^n$, then $\CC$ and $\CC'$ are equivalent.
		
		\begin{proof}
			For $\mu =\left(\begin{array}{c|c}
			\mu_1^{-1} &O \\
			\hline
			O & \mu_2			\end{array}\right) \in \M^{2n}$,
			$$(     I_n \mid  A )
			=(    I_n   \mid  \mu_1 B \mu_2)=(    \mu_1^{-1}   \mid  B \mu_2)
			= (      I_n \mid  B ) \mu. $$ Thus, $\CC$ and $\CC'$ are equivalent.
		\end{proof}    
	\end{proposition}

	\begin{definition}
		A matrix $A$ is called {\it symmetric} if $A^T=A$. If the matrix $A$ in a standard generator matrix $G=(I_n \mid A)$ of a self-dual code $\CC$ of length $2n$ over $GF(q)$ is symmetric, we call $G$ {\it a symmetric generator matrix of $\CC$}. If a self-dual code $\CC$ has a symmetric generator matrix, we call $\CC$ {\it a symmectric self-dual code}.
	\end{definition}
	
	\begin{definition} {\label{direct sum}}
		
		Let $\CC_1$, $\CC_2$ be self-dual codes of length $2l$ and $2m$ whose standard generator matrices are $(    I_l \mid  A_1 )$ and  $(    I_m \mid A_2 )$, respectively. {\it The direct sum} of two codes, $\CC_1  \oplus \CC_2$ is defined by the code having the generator matrix, 
		\begin{center}
			
			$  (    I_l \mid  A_1 ) \oplus (    I_m \mid  A_2 )= \left(\begin{array}{c|c|c|c}
			I_l &O& A_1 &O \\
			\hline
			O & I_m& O &A_2
			\end{array}\right)$.
		\end{center}
		
	\end{definition}

	\begin{corollary}\label{cor2.6}Let $I_n$ be the identity matrix of order $n$, $A$ is an $n\times n$ circulant matrix, $B$ is an $(n-1)\times (n-1)$ circulant matrix. Then,
		\begin{enumerate}
			\item a pure double circulant code over $GF(q)$ with a generator matrix of the form $$(I_n \mid A)$$ is equivalent to a code with symmetric generator matrix, and
			
			\item a bordered double circulant code  over $GF(q)$ with a generator matrix of the form \[
			\left(
			\begin{array}{ccc}
			& \alpha & \beta \cdots \beta \\ 
			\raisebox{-10pt}{{\large\mbox{{$I_n$}}}} & \beta& \raisebox{-15pt}{{\large\mbox{{$A$}}}} \\[-4ex]
			& \vdots & \\[-0.5ex]
			& \beta &
			\end{array}
			\right),
			\] where $\alpha$ and $\beta$ are elements in $GF(q)$, is equivalent to a code with symmetric generator matrix.
		\end{enumerate}

	\end{corollary} 
	\begin{proof}
		It is clear that a column reversed matrix of a circulant matrix $A$ is symmetric. Thus, the corollary follows directly from Proposition \ref{M-equivalent}.
	\end{proof}

	We remark that many MDS and optimal self-dual codes are obtained by using the construction method of pure double circulant codes and bordered double circulant codes in \cite{betsumiya2003,Gulliver2008}. These codes are all equivalent to codes with symmetric generator matrices.

	\section{Construction of symmetric self-dual codes}
	In this section, we introduce a construction method for symmetric self-dual codes over $GF(q)$ where $q \equiv 1 \pmod 4$. We also show that any symmetric self-dual code of length $2n+2$ is obtained from a symmetric self-dual code of length $2n$ by using this method. Thus, this is a complete method to obtain all symmetric self-dual codes. Our construction requires a square root of -1 in $GF(q)$; it is well-known that the equation $x^2 = -1$ has roots in $GF(q)$ if and only if $q \equiv 1 \pmod 4$. Thus, from now on, we assume that $q$ is a power of an odd prime such that $q \equiv 1 \pmod 4$. 
	
	\begin{lemma}\label{eigenspace}
		Let $\alpha$ be a root of -1 in $GF(q)$. If $\CC$ is a self-dual code of length $2n$ over $GF(q)$ with symmetric generator matrix
		$G=(    I_n \mid  A )$, then $A$ has an eigenvector $\x^T$ with eigenvalue $\alpha$ or $-\alpha$. 
		
		\begin{proof}
			Since $\CC$ is self-dual, $AA^T=-I$. With the assumption that $A$ is symmetric, we have that $A^2=-I$, and $$(A-\alpha I)(A + \alpha I)=A^2 +I =-I + I =O.$$ This implies that any non-zero vector $\x^T$ generated by column vectors of $A + \alpha I$, is an eigenvector of $A$ with eigenvalue $\alpha$ if $A \ne -\alpha I$. On the contrary, if $A = -\alpha I$, then it is obvious that any vector $\x^T$ in $GF(q)^n$ is an eigenvector of $A$ with eigenvalue $-\alpha$. Thus, the result follows.\end{proof}    
	\end{lemma}

	\begin{theorem}{\label{Building up}}
		Let 
		$(I_n \mid A)$
		be a generator matrix of symmetric self-dual code of length $2n$ over $GF(q)$ for $q \equiv 1 \pmod 4$. Let $\alpha$ be a square root of -1 and $\x^T$ be an non-zero eigenvector of $A$ corresponding eigenvalue $\alpha$, where $\x \x^T +1$ is a non-zero square in $GF(q)$. Take $\gamma$ be an element of $GF(q)$ satisfying $\gamma^2 = -1 - \x \x^T$ and $\gamma \ne \alpha$. And let $\beta=(\gamma-\alpha)^{-1}$ and $E= \beta \x^T \x$. 
		Then
		$$
		G'=	(I_{n+1} \mid A')=\left(\begin{array}{c|c|c|c}
		1 & O& \gamma & \x \\
		\hline
		
		O& I_n&\x^T& A+E\\
	\end{array}\right)$$
	is a generator matrix of symmetric self-dual code of length $2n+2$. In particular, if we take $\x$ a zero vector, then,
	$$
	G'=	(1 \mid \alpha)\oplus(I_{n} \mid A) =\left(\begin{array}{c|c|c|c}
		1 & O& \alpha & O \\
		\hline
		
		O& I_n&O& A\\
	\end{array}\right)$$
	is a generator matrix of symmetric self-dual code of length $2n+2$ with minimum weight two.
	\end{theorem}

	\begin{proof}
		The `particular' part is trivial. Since the row rank of $G'$ is $n+1$, we have only to show that $A' (A')^T$ is equal to $-I_{n+1}$.
		
		By the assumption, we have that $A A^T = -I_n$ and  $A\x^T =  \alpha \x^T$, thus $AE^T = A(\beta\x^T\x)=\beta(A\x^T)\x=\alpha\beta \x^T\x$ and $EA^T =(AE^T)^T=(\alpha\beta \x^T\x)^T=\alpha\beta \x^T\x.$ Therefore,
		
		\begin{eqnarray*}
			A' (A')^T&= &\left(\begin{array}{c|c}
				\gamma & \x  \\
				\hline
				
				\x^T& A+E \\
			\end{array}\right)\left(\begin{array}{c|c}
				\gamma & \x  \\
				\hline
				
				\x^T& A+E \\
			\end{array}\right)^T\\
			&=&\left(\begin{array}{c|c}
				\gamma^2+ \x\x^T & \gamma \x +  \x A^T +\x E^T  \\
				\hline
				
				\gamma \x^T +  A \x^T +E \x^T  &\x^T\x+ AA^T + AE^T + E A^T + EE^T  \\
			\end{array}\right)\\
			&=&\left(\begin{array}{c|c}
				-1 & \gamma \x +  \alpha \x +\beta\x (\x^T\x)^T  \\
				\hline
				
				\gamma \x^T +  A \x^T +E \x^T  & -I_n+\x^T\x+ 2\alpha\beta \x^T\x + EE^T  \\
			\end{array}\right).\\
		\end{eqnarray*}
		Since $\x \x^T = -\gamma^2  -1$, we simplify the (1,2)-block matrix as
		\begin{eqnarray*} \gamma \x +  \alpha \x +\beta\x (\x^T\x)^T 
			&=&  \gamma \x +  \alpha \x + \beta(-\gamma^2-1)\x \\ 
			&=&  (\gamma  +  \alpha -  \beta(\gamma^2+1))\x \\ 
			&=& \beta( \beta^{-1} ( \gamma  +  \alpha) -(\gamma^2+1))\\ 
			&=& \beta(( \gamma  -  \alpha) ( \gamma  +  \alpha) -(\gamma^2+1))\\ 
			&=& \beta(( \gamma^2+1) -(\gamma^2+1))\\ 
			&=&O_{1\times n}.	\end{eqnarray*}	
		The (2,1)-block matrix $\gamma \x^T +  A \x^T +E \x^T=O_{n \times 1}$ since this is the transpose of the (1,2)-block matrix.
		Finally, there remains only to show that the (2,2)-block matrix is equal to $-I_n$. Recall that $\alpha^2=-1$ and $\beta=(\gamma-\alpha)^{-1}$. Thus,
		\begin{eqnarray*}  \x^T\x+ 2\alpha\beta \x^T\x + EE^T
			&=& \x^T\x  + 2\alpha\beta \x^T\x+ \beta^{2}(\x^T\x)(\x^T\x)^T \\
			&=&  \x^T\x + 2\alpha\beta \x^T\x+ \beta^{2} \x^T(-\gamma^2-1)\x \\
			&=&  (1+2\alpha\beta - \beta^{2}\gamma^2-\beta^{2})\x^T\x \\
			&=& \beta^2( \beta^{-2} + 2\alpha\beta^{-1} -\gamma^2-1) \x^T\x \\
			&=& \beta^2\{ ( \gamma  -  \alpha)^2 + 2 \alpha ( \gamma  -  \alpha)  -\gamma^2-1)\} \x^T\x\\ 
			&=& \beta^2( \gamma^2 - 2 \gamma \alpha -1 +2\gamma \alpha +2 - \gamma^2 -1) \x^T\x\\ 
			&=&O_{n \times n}	\end{eqnarray*}	
		and the (2,2)-block matrix is equal to $-I_n$. This is what was to be shown.
	\end{proof}

	By the construction method of Theorem \ref{Building up}, we obtain symmetric self-dual codes of length $2n+2$ from a symmetric self-dual code of length $2n$.
	From now on, we discuss the converse of Theorem \ref{Building up}.

	\begin{lemma}\label{eigenvector}
		Suppose that $\CC$ is a symmetric self-dual code over $GF(q)$ with generator matrix in the form:
		$$
		\left(\begin{array}{c|c|c|c}
		I_n & O & \gamma & \x \\
		\hline
		
		O& 1&\x^T&A\\
	\end{array}\right),$$
	where $\x$ is non-zero. Let $\alpha$ be a square root of -1 over a finite field $GF(q)$ which is not equal to $\gamma$ and let $\beta  =(\gamma-\alpha)^{-1}$. Then $\x$ is an eigenvector of $A-\beta \x^T \x$ with eigenvalue $\alpha$.
	\end{lemma}
	
	\begin{proof}
		Since $\CC$ is a symmetric self-dual code, 
		
		$$\left(\begin{array}{c|c}
		\gamma & \x  \\
		\hline
		
		\x^T& A \\
	\end{array}\right)\left(\begin{array}{c|c}
		\gamma & \x  \\
		\hline
		
		\x^T& A \\
	\end{array}\right)^T=-I_{n+1}.$$ Thus, 
	\begin{equation}\label{sim_eq}
	\begin{cases} \gamma^2 + \x \x^T = -1 \\  \gamma \x + \x A^T = O  \\\gamma \x^T +  A\x^T = O \\\x^T\x +AA^T =-I_n.
	\end{cases}
	\end{equation}
	By using these equalties, we show that 
	\begin{eqnarray*} (A- \beta \x^T\x) \x^T &=& A \x^T - \beta \x^T( \x\x^T) \\
	&=& -\gamma \x^T -\beta\x^T(-1-\gamma^2)  \\
	&=& \beta(-\beta^{-1} \gamma +1 + \gamma^2) \x^T\\ 
	&=& \beta(-(\gamma-\alpha) \gamma +1 + \gamma^2) \x^T\\ 
	&=& \beta( \alpha \gamma +1 ) \x^T\\ 
	&=& (\gamma-\alpha)^{-1}( \alpha \gamma -\alpha^2) \x^T\\ 
	&=&\alpha \x^T.	\end{eqnarray*}	
	Thus the result follows.
	\end{proof}
	
	\begin{theorem}{\label{converse}}
		
		Any symmetric self-dual code $\CC$ of length $2n$ over $GF(q)$ for a prime $q=4k+1$ can be constructed from some symmetric self-dual code $\CC'$ of length $2n-2$ by the construction method in Theorem \ref{Building up}.

	\end{theorem}
	\begin{proof}
		This follows directly from Lemma \ref{eigenvector}.
	\end{proof}

	\begin{remark}    
		Theorem \ref{Building up} and \ref{converse} might be regarded as a special case of well-known `building-up' construction method \cite[Proposition 2.1, 2.2]{Kim2004}. But Theorem \ref{Building up} and \ref{converse}  has a significant differences. We only have to choose vectors from an eigenspace of $A$ with an eigenvalue of a root of $-1$. This improves the efficiency to find the best self-dual code from a self-dual code of smaller length. We also point out that all of the self-dual codes used in this method have symmetric generator matrices. Thus, we can focus our concern in one subclass of self-dual codes that have a certain automorphism in their automorphism group.
		
	\end{remark}

	\begin{example}\label{ex3.6}
		Let $\CC_5^{16}$ be a symmetric self-dual [16,8,6] code over $GF(5)$ with generator matrix 
		$$G= (I_8 \mid A)=\begin{smatrix}
		1& 0& 0& 0& 0& 0& 0& 0& 1& 4& 3& 3& 2& 4& 0& 2\\
		0& 1& 0& 0& 0& 0& 0& 0& 4& 2& 0& 2& 4& 3& 2& 1\\
		0& 0& 1& 0& 0& 0& 0& 0& 3& 0& 1& 1& 3& 3& 3& 4\\
		0& 0& 0& 1& 0& 0& 0& 0& 3& 2& 1& 0& 2& 0& 0& 1\\
		0& 0& 0& 0& 1& 0& 0& 0& 2& 4& 3& 2& 4& 1& 3& 0\\
		0& 0& 0& 0& 0& 1& 0& 0& 4& 3& 3& 0& 1& 1& 2& 2\\
		0& 0& 0& 0& 0& 0& 1& 0& 0& 2& 3& 0& 3& 2& 3& 2\\
		0& 0& 0& 0& 0& 0& 0& 1& 2& 1& 4& 1& 0& 2& 2& 3
		\end{smatrix},$$ which is optimal.
		Then, the eigenspace of $A$ with eigenvalue $\alpha=2$ is a subspace of $GF(5)^8$ of dimension four generated by row vectors of the matrix 
		$\begin{smatrix}
		1& 0& 0& 0& 0& 3& 0& 2\\
		0& 1& 0& 0& 3& 4& 2& 2\\
		0& 0& 1& 0& 4& 3& 2& 1\\
		0& 0& 0& 1& 1& 0& 2& 0
		\end{smatrix}.$
		Among these $5^4=625$ eigenvectors, if we choose the vector $\x=43411113$, then using the construction method in Theorem \ref{Building up} with $\gamma=0$ and $\beta=(\gamma-\alpha)^{-1}=2$, we obtain an `optimal' symmetric self-dual [18,9,7] code with generator matrix
		
		$$
		G'=	\left(\begin{array}{c|c|c|c}
		1 & O& \gamma & \x \\
		\hline
		
		O& I_n&\x^T& A+\beta \x^T \x \\
	\end{array}\right)= \begin{smatrix}
		1& 0& 0& 0& 0& 0& 0& 0& 0& 0& 4& 3& 4& 1& 1& 1& 1& 3\\
		0& 1& 0& 0& 0& 0& 0& 0& 0& 4& 3& 3& 0& 1& 0& 2& 3& 1\\
		0& 0& 1& 0& 0& 0& 0& 0& 0& 3& 3& 0& 4& 3& 0& 4& 3& 4\\
		0& 0& 0& 1& 0& 0& 0& 0& 0& 4& 0& 4& 3& 4& 1& 1& 1& 3\\
		0& 0& 0& 0& 1& 0& 0& 0& 0& 1& 1& 3& 4& 2& 4& 2& 2& 2\\
		0& 0& 0& 0& 0& 1& 0& 0& 0& 1& 0& 0& 1& 4& 1& 3& 0& 1\\
		0& 0& 0& 0& 0& 0& 1& 0& 0& 1& 2& 4& 1& 2& 3& 3& 4& 3\\
		0& 0& 0& 0& 0& 0& 0& 1& 0& 1& 3& 3& 1& 2& 0& 4& 0& 3\\
		0& 0& 0& 0& 0& 0& 0& 0& 1& 4& 1& 4& 3& 2& 1& 3& 3& 1
	\end{smatrix}$$
	
\end{example}    

We close this section comparing the complexity of our method with that of the well-known `building-up' method in \cite[Proposition 2.1]{Kim2004}. If we apply `building-up' method in \cite[Proposition 2.1]{Kim2004} to the self-dual code $\CC_5^{16}$ of length 16 in Example \ref{ex3.6} to construct self-dual codes of length 18, a vector is typically chosen in $GF(5)^{15}$, i.e., there are $5^{15}$ possible choices. In contrast, as we have already seen in Example \ref{ex3.6}, the number of possible choices of vectors is reduced only to $5^4$ when our new method is applied.

According to our computational experiences to obtain best self-dual codes in Table \ref{our_results1}, it needs only about $q^{\lfloor \frac{n}{2}\rfloor}$ choices of eigenvectors when given length is $2n$. Due to this reduced complexity, we succeed in constructing self-dual codes of length greater than 22.

\section{Computational results of optimal or best-known self-dual codes}
In this section, we construct optimal self-dual codes over $GF(13)$ and $GF(17)$ by using the method in the previous section. From now on, for the brevity, we denote a symmetric $[2n,k,d]$ self-dual code over $GF(p)$ as $\CC_{p}^{2n}$ and its generator matrix as $(I_n \mid A_{p}^{2n})$. All the computations are done in Magma \cite{Magma}.

\subsection{Optimal self-dual codes over $GF(13)$. }


In \cite{betsumiya2003}, the optimal minimum weights of self-dual codes over $GF(13)$ are determined for lengths up to 20 except 12, and the minimum optimal weight of length 12 is determined in \cite{grassl2008}. However, we pointed out that the existence of optimal self-dual codes of length 18 turns out to be unknown. This is to be discussed in Remark \ref{remark1}. We obtain [18,9,8] self-dual code, which is now known to have the best-known minimum weight, with a symmetric generator matrix,

$$G_{13}^{18}= \begin{smatrix}
1 & 0 & 0 & 0 & 0 & 0 & 0 & 0 & 0 & 10 & 5 & 5 & 0 & 1 & 9 & 12 & 2 & 3 \\
0 & 1 & 0 & 0 & 0 & 0 & 0 & 0 & 0 & 5 & 7 & 11 & 10 & 4 & 4 & 12 & 6 & 5 \\
0 & 0 & 1 & 0 & 0 & 0 & 0 & 0 & 0 & 5 & 11 & 5 & 3 & 5 & 3 & 7 & 6 & 5 \\
0 & 0 & 0 & 1 & 0 & 0 & 0 & 0 & 0 & 0 & 10 & 3 & 5 & 6 & 6 & 0 & 6 & 2 \\
0 & 0 & 0 & 0 & 1 & 0 & 0 & 0 & 0 & 1 & 4 & 5 & 6 & 0 & 10 & 5 & 1 & 9 \\
0 & 0 & 0 & 0 & 0 & 1 & 0 & 0 & 0 & 9 & 4 & 3 & 6 & 10 & 12 & 9 & 4 & 6 \\
0 & 0 & 0 & 0 & 0 & 0 & 1 & 0 & 0 & 12 & 12 & 7 & 0 & 5 & 9 & 3 & 12 & 1 \\
0 & 0 & 0 & 0 & 0 & 0 & 0 & 1 & 0 & 2 & 6 & 6 & 6 & 1 & 4 & 12 & 4 & 10 \\
0 & 0 & 0 & 0 & 0 & 0 & 0 & 0 & 1 & 3 & 5 & 5 & 2 & 9 & 6 & 1 & 10 & 11
\end{smatrix}.$$



In Table \ref{Table_GF(13)}, we illustrate the chain of self-dual codes constructed by using Theorem \ref{Building up}, successively from [26,13,10] code $\CC_{13}^{26,1}$ to [40,20,14] code $\CC_{13}^{40,1}$. These self-dual codes are all new and have the best-known minimum weights. The [26,13,10] self-dual code $\CC_{13}^{26,1}$ has a generator matrix $(I_{13} \mid A_{13}^{26,1})$ where
$$A_{13}^{26, 1}= \begin{smatrix}
7 & 7 & 1 & 8 & 3 & 6 & 3 & 8 & 10 & 10 & 10 & 0 & 9 \\
7 & 8 & 10 & 8 & 7 & 5 & 7 & 8 & 8 & 11 & 7 & 0 & 4 \\
1 & 10 & 11 & 11 & 10 & 9 & 5 & 7 & 10 & 4 & 8 & 7 & 11 \\
8 & 8 & 11 & 12 & 7 & 11 & 3 & 12 & 4 & 12 & 11 & 8 & 11 \\
3 & 7 & 10 & 7 & 10 & 0 & 8 & 12 & 12 & 7 & 10 & 10 & 1 \\
6 & 5 & 9 & 11 & 0 & 8 & 5 & 7 & 3 & 11 & 8 & 4 & 8 \\
3 & 7 & 5 & 3 & 8 & 5 & 3 & 4 & 11 & 5 & 6 & 11 & 6 \\
8 & 8 & 7 & 12 & 12 & 7 & 4 & 8 & 0 & 4 & 3 & 1 & 9 \\
10 & 8 & 10 & 4 & 12 & 3 & 11 & 0 & 4 & 8 & 3 & 10 & 7 \\
10 & 11 & 4 & 12 & 7 & 11 & 5 & 4 & 8 & 5 & 9 & 1 & 4 \\
10 & 7 & 8 & 11 & 10 & 8 & 6 & 3 & 3 & 9 & 11 & 0 & 8 \\
0 & 0 & 7 & 8 & 10 & 4 & 11 & 1 & 10 & 1 & 0 & 5 & 4 \\
9 & 4 & 11 & 11 & 1 & 8 & 6 & 9 & 7 & 4 & 8 & 4 & 10
\end{smatrix}.$$
We give generator matrices of new symmetric self-dual codes over $GF(13)$ of lengths upto 40  in Appendix \ref{GF(13)}.
\begin{table}[h!]
\begin{center}
	\begin{tabular}{lllll}
		\hline
		Code &$\alpha$ & $\gamma$   & $\x$  & min. wt.   \\\hline
		\hline
		$\CC_{13}^{26,1}$ &  	&              & 		&	10         	\\\hline
		$\CC_{13}^{28,1}$ &  8	&    4          & $ \begin{smatrix}2, 10, 8, 6, 3, 1, 12, 1,11, 8, 9, 11, 2\end{smatrix}$			&	  11        	\\\hline
		$\CC_{13}^{30,1}$  &  8	&    11          & $ \begin{smatrix}10, 8, 9, 2, 1, 4, 12, 12, 7,
		12, 2, 2, 6, 6\end{smatrix}$			&	  11        	\\\hline
		$\CC_{13}^{32,1}$ &  8	&    11         &$  \begin{smatrix} 5, 8, 5, 2, 7, 11, 11, 10,
		12, 2, 11, 12, 3, 4, 7\end{smatrix}$			&	 12         	\\\hline
		$\CC_{13}^{34,1}$  & 5 	&    1         & $ \begin{smatrix}0, 3, 7, 5, 1, 10, 11,
		3, 7, 2, 10, 12, 2, 6, 12, 10\end{smatrix}$			&	   12       	\\\hline
		$\CC_{13}^{36,1}$  &  8	&    6         & $ \begin{smatrix}3, 1, 1, 5, 8, 1, 6, 
		3, 1, 4, 1, 1, 3, 11, 8, 2, 4\end{smatrix}$			&	   13       	\\\hline
		$\CC_{13}^{38,1}$  &  5 	&    3          & $ \begin{smatrix}8, 0, 3, 2, 11, 6,
		8, 3, 9, 3, 7, 1, 7, 2, 8, 11, 9, 2\end{smatrix}$			&	     13     	\\\hline
		$\CC_{13}^{40,1}$  &  5	&    8          &$ \begin{smatrix}5, 10, 5, 4, 1,
		8, 1, 2, 3, 4, 11, 5, 8, 6, 3, 2, 12, 9, 3\end{smatrix}$			&	      14    	\\\hline
		
	\end{tabular}
	\caption{Constuction of a chain of best-known self-dual codes over GF(13) }
	\label{Table_GF(13)}
\end{center}	
\end{table}

\subsection{Optimal self-dual codes over $GF(17)$  }

We consruct [26,13,10] and [28,14,11] self-dual code over GF(17) which are new, succesively from [24,12,9] self-dual code by using Theorem \ref{Building up} as follows.
At first, we obtain [24,12,9] code with generator matrix $(I_{12} \mid A_{17}^{24,1})$ where
$$A_{17}^{24, 1}= \begin{smatrix}
10 & 8 & 15 & 7 & 4 & 13 & 10 & 11 & 6 & 12 & 5 & 2 \\
8 & 3 & 5 & 14 & 15 & 14 & 0 & 6 & 12 & 8 & 9 & 9 \\
15 & 5 & 13 & 1 & 9 & 0 & 6 & 9 & 14 & 3 & 8 & 9 \\
7 & 14 & 1 & 2 & 3 & 15 & 6 & 5 & 14 & 0 & 12 & 10 \\
4 & 15 & 9 & 3 & 15 & 2 & 2 & 12 & 12 & 14 & 9 & 14 \\
13 & 14 & 0 & 15 & 2 & 9 & 3 & 2 & 13 & 8 & 0 & 8 \\
10 & 0 & 6 & 6 & 2 & 3 & 7 & 14 & 4 & 2 & 0 & 5 \\
11 & 6 & 9 & 5 & 12 & 2 & 14 & 12 & 3 & 15 & 13 & 16 \\
6 & 12 & 14 & 14 & 12 & 13 & 4 & 3 & 7 & 1 & 5 & 0 \\
12 & 8 & 3 & 0 & 14 & 8 & 2 & 15 & 1 & 5 & 13 & 13 \\
5 & 9 & 8 & 12 & 9 & 0 & 0 & 13 & 5 & 13 & 10 & 12 \\
2 & 9 & 9 & 10 & 14 & 8 & 5 & 16 & 0 & 13 & 12 & 1
\end{smatrix}.$$

By taking $\gamma = 4$ and the eigenvector $(5 , 11 , 16 , 1 , 11 , 8 , 3 , 4 , 8 , 4 , 6 , 6)$ of $A_{17}^{12,9}$ corresponding eigenvalue $\alpha=13$, we obtain [26,13,10] self-dual code with generator matrix $(I_{13} \mid A_{17}^{26,1})$ where
$$A_{17}^{26, 1}= \begin{smatrix}
4 & 5 & 11 & 16 & 1 & 11 & 8 & 3 & 4 & 8 & 4 & 6 & 6 \\
5 & 11 & 0 & 8 & 14 & 13 & 1 & 14 & 5 & 11 & 6 & 13 & 10 \\
11 & 0 & 16 & 10 & 9 & 11 & 8 & 2 & 3 & 6 & 5 & 13 & 13 \\
16 & 8 & 10 & 11 & 3 & 14 & 16 & 12 & 0 & 13 & 11 & 3 & 4 \\
1 & 14 & 9 & 3 & 0 & 15 & 16 & 0 & 14 & 15 & 9 & 0 & 15 \\
11 & 13 & 11 & 14 & 15 & 11 & 13 & 4 & 9 & 6 & 11 & 13 & 1 \\
8 & 1 & 8 & 16 & 16 & 13 & 0 & 6 & 6 & 4 & 12 & 6 & 14 \\
3 & 14 & 2 & 12 & 0 & 4 & 6 & 6 & 7 & 7 & 12 & 15 & 3 \\
4 & 5 & 3 & 0 & 14 & 9 & 6 & 7 & 14 & 7 & 0 & 16 & 2 \\
8 & 11 & 6 & 13 & 15 & 6 & 4 & 7 & 7 & 15 & 5 & 11 & 6 \\
4 & 6 & 5 & 11 & 9 & 11 & 12 & 12 & 0 & 5 & 7 & 16 & 16 \\
6 & 13 & 13 & 3 & 0 & 13 & 6 & 15 & 16 & 11 & 16 & 6 & 8 \\
6 & 10 & 13 & 4 & 15 & 1 & 14 & 3 & 2 & 6 & 16 & 8 & 14
\end{smatrix}.$$

Again, by taking $\gamma = 4$ and the eigenvector $(14 , 11 , 12 , 0 ,11 , 11 , 0 ,10 , 12 , 15 , 11 , 0 , 4 )$ of $A_{17}^{26, 1}$ corresponding eigenvalue $\alpha=13$, we obtain [28,14,11] self-dual code with generator matrix $(I_{14} \mid A_{17}^{28, 1})$ where
[28,14,11] self-dual code:\\
$$A_{17}^{28, 1}= \begin{smatrix}
4 & 14 & 11 & 12 & 0 & 11 & 11 & 0 & 10 & 12 & 15 & 11 & 0 & 4 \\
14 & 3 & 3 & 15 & 16 & 16 & 9 & 8 & 12 & 8 & 13 & 2 & 6 & 13 \\
11 & 3 & 7 & 8 & 8 & 10 & 9 & 1 & 15 & 13 & 4 & 2 & 13 & 7 \\
12 & 15 & 8 & 0 & 10 & 0 & 2 & 8 & 0 & 4 & 3 & 13 & 13 & 2 \\
0 & 16 & 8 & 10 & 11 & 3 & 14 & 16 & 12 & 0 & 13 & 11 & 3 & 4 \\
11 & 16 & 10 & 0 & 3 & 13 & 11 & 16 & 1 & 5 & 8 & 5 & 0 & 12 \\
11 & 9 & 9 & 2 & 14 & 11 & 7 & 13 & 5 & 0 & 16 & 7 & 13 & 15 \\
0 & 8 & 1 & 8 & 16 & 16 & 13 & 0 & 6 & 6 & 4 & 12 & 6 & 14 \\
10 & 12 & 15 & 0 & 12 & 1 & 5 & 6 & 10 & 5 & 13 & 13 & 15 & 8 \\
12 & 8 & 13 & 4 & 0 & 5 & 0 & 6 & 5 & 15 & 4 & 8 & 16 & 8 \\
15 & 13 & 4 & 3 & 13 & 8 & 16 & 4 & 13 & 4 & 7 & 15 & 11 & 5 \\
11 & 2 & 2 & 13 & 11 & 5 & 7 & 12 & 13 & 8 & 15 & 3 & 16 & 13 \\
0 & 6 & 13 & 13 & 3 & 0 & 13 & 6 & 15 & 16 & 11 & 16 & 6 & 8 \\
4 & 13 & 7 & 2 & 4 & 12 & 15 & 14 & 8 & 8 & 5 & 13 & 8 & 16
\end{smatrix}.$$

In Table \ref{Table_GF(17)}, we illustrate a chain of self-dual codes constructed by using Theorem \ref{Building up}, successively from a [28,14,10] code to a [40,20,14] code.
The [28,14,10] self-dual code $\CC_{17}^{28,2}$ has a generator matrix $(I_{14} \mid A_{17}^{28,2})$ where
$$A_{17}^{28, 2}= \begin{smatrix}
4 & 2 & 4 & 9 & 9 & 7 & 16 & 7 & 13& 4 & 14 & 11 & 1 & 7 \\
2 & 14 & 16 & 14 & 12 & 3 & 1 & 0 & 3& 0 & 5 & 3 & 4 & 16 \\
4 & 16 & 4 & 2 & 0 & 5 & 16 & 13 & 2& 3 & 12 & 9 & 16 & 2 \\
9 & 14 & 2 & 16 & 12 & 0 & 15 & 14 & 8& 16 & 7 & 14 & 11 & 9 \\
9 & 12 & 0 & 12 & 12 & 7 & 0 & 4 & 13& 2 & 10 & 1 & 9 & 1 \\
7 & 3 & 5 & 0 & 7 & 13 & 12 & 5 & 2& 7 & 14 & 5 & 2 & 13 \\
16 & 1 & 16 & 15 & 0 & 12 & 4 & 14 & 11& 8 & 9 & 8 & 11 & 1 \\
7 & 0 & 13 & 14 & 4 & 5 & 14 & 13 & 8& 11 & 5 & 8 & 16 & 3 \\
13 & 3 & 2 & 8 & 13 & 2 & 11 & 8 & 14& 9 & 12 & 9 & 9 & 6 \\
4 & 0 & 3 & 16 & 2 & 7 & 8 & 11 & 9& 3 & 1 & 16 & 10 & 11 \\
14 & 5 & 12 & 7 & 10 & 14 & 9 & 5 & 12& 1 & 7 & 4 & 14 & 1 \\
11 & 3 & 9 & 14 & 1 & 5 & 8 & 8 & 9& 16 & 4 & 6 & 11 & 4 \\
1 & 4 & 16 & 11 & 9 & 2 & 11 & 16 & 9& 10 & 14 & 11 & 15 & 1 \\
7 & 16 & 2 & 9 & 1 & 13 & 1 & 3 & 6& 11 & 1 & 4 & 1 & 11
\end{smatrix}.$$

\begin{table}[h]
\begin{center}
	\begin{tabular}{lllll}
		\hline
		Code &$\alpha$ & $\gamma$   & $\x$  & min. wt.   \\\hline
		\hline
		$\CC_{17}^{28,2}$ &  	&            & 		&	  10        	\\\hline
		$\CC_{17}^{30,1}$  & 13 	&    14          & $ \begin{smatrix}14, 14, 0, 0, 15, 9, 9, 8, 1,
		12, 1, 2, 8, 15\end{smatrix}$			&	  12        	\\\hline
		$\CC_{17}^{32,1}$ &  4	&    11         &$  \begin{smatrix} 9, 4, 10, 11, 6, 4, 0, 9, 
		7, 7, 14, 4, 15, 13, 7\end{smatrix}$			&	 12         	\\\hline
		$\CC_{17}^{34,1}$  &  4 	&    1        & $ \begin{smatrix}3, 16, 5, 0, 0, 0, 11, 
		7, 7, 0, 6, 6, 5, 7, 2, 11,\end{smatrix}$			&	   12       	\\\hline
		$\CC_{17}^{36,1}$  &  4 &    7         & $ \begin{smatrix}10, 4, 7, 7, 6, 14, 
		9, 5, 6, 9, 8, 14, 13, 7, 4, 6, 14\end{smatrix}$			&	   13       	\\\hline
		$\CC_{17}^{38,1}$  &  13 	&    4          & $ \begin{smatrix}1, 9, 8, 8, 10, 7,
		13, 1, 9, 1, 10, 9, 0, 10, 16, 5, 2, 9\end{smatrix}$			&	     14     	\\\hline
		$\CC_{17}^{40,1}$  &  4&    9         &$ \begin{smatrix}12, 9, 13, 3, 
		0, 3, 0, 12, 15, 16, 3, 6, 15, 6, 15, 13, 10, 10, 2\end{smatrix}$			&	      14    	\\\hline
		
	\end{tabular}
	\caption{Constuction of a chain of best-known self-dual codes over GF(17) }
	\label{Table_GF(17)}
\end{center}	
\end{table}
We give generator matrices of new self-dual codes over $GF(17)$ of length upto 40 in Appendix \ref{GF(17)}.
Additionally, we constructed the best [34,17,13] self-dual code $\CC_{17}^{34,2}$ with generator matrix $(I_{17} \mid A_{17}^{34,2})$ where
$$A_{17}^{34, 2}= \begin{smatrix}
3 & 1 & 3 & 1 & 3 & 10 & 14 & 6 & 2
& 9 & 14 & 15 & 10 & 8 & 16 & 2 & 0 \\
1 & 5 & 6 & 3 & 7 & 1 & 9 & 15 & 15
& 0 & 10 & 9 & 4 & 13 & 16 & 11 & 7 \\
3 & 6 & 3 & 12 & 16 & 11 & 2 & 15 & 3
& 14 & 6 & 13 & 11 & 12 & 13 & 1 & 4 \\
1 & 3 & 12 & 3 & 11 & 16 & 7 & 0 & 12
& 15 & 0 & 9 & 3 & 11 & 2 & 1 & 10 \\
3 & 7 & 16 & 11 & 13 & 7 & 2 & 3 & 4
& 9 & 3 & 7 & 4 & 15 & 7 & 5 & 6 \\
10 & 1 & 11 & 16 & 7 & 7 & 14 & 13 & 2
& 1 & 5 & 1 & 14 & 5 & 8 & 8 & 15 \\
14 & 9 & 2 & 7 & 2 & 14 & 8 & 8 & 6
& 1 & 9 & 6 & 3 & 9 & 5 & 5 & 13 \\
6 & 15 & 15 & 0 & 3 & 13 & 8 & 5 & 11
& 14 & 3 & 3 & 14 & 2 & 16 & 7 & 2 \\
2 & 15 & 3 & 12 & 4 & 2 & 6 & 11 & 5
& 7 & 15 & 15 & 10 & 7 & 8 & 2 & 12 \\
9 & 0 & 14 & 15 & 9 & 1 & 1 & 14 & 7
& 2 & 0 & 10 & 4 & 15 & 9 & 13 & 6 \\
14 & 10 & 6 & 0 & 3 & 5 & 9 & 3 & 15
& 0 & 16 & 9 & 16 & 0 & 14 & 4 & 3 \\
15 & 9 & 13 & 9 & 7 & 1 & 6 & 3 & 15
& 10 & 9 & 10 & 5 & 9 & 2 & 7 & 3 \\
10 & 4 & 11 & 3 & 4 & 14 & 3 & 14 & 10
& 4 & 16 & 5 & 9 & 12 & 2 & 7 & 2 \\
8 & 13 & 12 & 11 & 15 & 5 & 9 & 2 & 7
& 15 & 0 & 9 & 12 & 3 & 10 & 15 & 13 \\
16 & 16 & 13 & 2 & 7 & 8 & 5 & 16 & 8
& 9 & 14 & 2 & 2 & 10 & 11 & 16 & 10 \\
2 & 11 & 1 & 1 & 5 & 8 & 5 & 7 & 2 &
13 & 4 & 7 & 7 & 15 & 16 & 15 & 5 \\
0 & 7 & 4 & 10 & 6 & 15 & 13 & 2 & 12
& 6 & 3 & 3 & 2 & 13 & 10 & 5 & 14
\end{smatrix}.$$

\subsection{Quadratic residue codes over $GF(q)$. }
In addition to our results of self-dual codes over $GF(13)$ and $GF(17)$, we want to construct self-dual codes over other finite fields. In \cite{betsumiya2003}, it is reported that some optimal self-dual codes are obtained from quadratic residue codes following \cite[Theorem 15]{deboer1996}. We also obtain new quadratic residue codes in the following theorem. Among them, $[32, 16, 14]$ code over $GF(19)$, $[20, 10, 10]$ code over $GF(23)$,  $[24, 12, 12]$ code over $GF(29)$, and $[24, 12, 12]$ over $GF(41)$ give the best-known minimum weights which were unknown so far. The new results are updated in Table \ref{former_result_B} and \ref{former_result_C}, and their generator matrices are given in Appendix \ref{QR codes}.

\begin{theorem}{\label{QR}}
The following quadratic residue codes are self-dual:

\begin{tabular}{ll}
	$[24, 12, 10]$ code over $GF(13)$, 	& $[32, 16, 14]$ code over $GF(19)$,\\
	$[20, 10, 10]$ code over $GF(23)$,  &
	$[24, 12, 12]$ code over $GF(29)$, \\
	$[24, 12, 12]$ code over $GF(31)$,&
	$[24, 12, 12]$ code over $GF(41)$,\\
	$[32, 16, 14]$ code over $GF(41)$. &\\
\end{tabular}

%

\end{theorem}

\begin{remark} \label{remark1}   
The $[18,9,9]$ linear code, quadratic residue code over $GF(13)$ of length 18, is reported as an optimal self-dual code of that parameter in \cite{betsumiya2003} referring \cite[Theorem 15]{deboer1996}. But we point out that the quadratic residue code over $GF(13)$ of length 18 is not self-dual, which have a generator matrix in the standard form $(I \mid A)$ where
$$A= \begin{smatrix}
1 & 8 & 10 & 11 & 4 & 11 & 10 & 8 & 4 \\
5 & 2 & 6 & 0 & 5 & 7 & 9 & 11 & 11 \\
2 & 8 & 9 & 2 & 8 & 1 & 1 & 12 & 6 \\
1 & 10 & 5 & 7 & 6 & 6 & 11 & 9 & 10 \\
4 & 7 & 11 & 10 & 10 & 11 & 7 & 4 & 0 \\
9 & 11 & 6 & 6 & 7 & 5 & 10 & 1 & 10 \\
12 & 1 & 1 & 8 & 2 & 9 & 8 & 2 & 6 \\
11 & 9 & 7 & 5 & 0 & 6 & 2 & 5 & 11 \\
8 & 10 & 11 & 4 & 11 & 10 & 8 & 1 & 4
\end{smatrix}.$$

For the details of the self-duality of quadratic residue codes, we refer \cite[Chap. 6.6]{HP3}. Theorem 6.6.18 in \cite{HP3} implies that quadratic residue code over $GF(13)$ of length 18 is an iso-dual code, i.e., the code is equivalent to its dual. Therefore, the existence of optimal self-dual code over $GF(13)$ of length 18 turns out unknown, and that is the reason why we inscribe the `?' in Table 4.
\end{remark}    

\begin{remark}    
We also point out that the quadratic residue code over $GF(17)$ of length 14 is MDS but isodual code with a generator matrix in the standard form $(I \mid A)$ where
$$A=\begin{smatrix}
1 & 5 & 2 & 4 & 2 & 5 & 10 \\
12 & 10 & 12 & 16 & 11 & 11 & 11 \\
6 & 8 & 5 & 2 & 11 & 7 & 3 \\
10 & 5 & 11 & 11 & 5 & 10 & 1 \\
7 & 11 & 2 & 5 & 8 & 6 & 3 \\
11 & 11 & 16 & 12 & 10 & 12 & 11 \\
5 & 2 & 4 & 2 & 5 & 1 & 10
\end{smatrix}.$$

\end{remark}

\section{Conclusions}

In this paper, we introduced a new construction method of symmetric self-dual codes. Using this construction method, we have constructed many new self-dual codes. We also obtained new quadratic residue codes. Consequently, we improved the bounds of the highest minimum weights of self-dual codes over some finite fields, which stayed unknown for almost two decades because of their computational complexity. Our computational results give twenty new highest minimum weights of self-dual codes and 2967 new self-dual codes up to equivalence. The highest minimum weights of self-dual over $GF(q)$ where $q \equiv 3 \pmod 4$ will be improved in our following works.





\appendix

\section{\\Self-dual codes over $GF(13)$}\label{GF(13)}



\begin{itemize}
\item Symmetric self-dual [26,13,10] code over $GF(13)$ with $(I_{13} \mid A_{13}^{26,1})$ where \\
$$A_{13}^{26,1}= \begin{smatrix}
7 & 7 & 1 & 8 & 3 & 6 & 3 & 8 & 10 &
10 & 10 & 0 & 9 \\
7 & 8 & 10 & 8 & 7 & 5 & 7 & 8 & 8 &
11 & 7 & 0 & 4 \\
1 & 10 & 11 & 11 & 10 & 9 & 5 & 7 & 10
& 4 & 8 & 7 & 11 \\
8 & 8 & 11 & 12 & 7 & 11 & 3 & 12 & 4
& 12 & 11 & 8 & 11 \\
3 & 7 & 10 & 7 & 10 & 0 & 8 & 12 & 12
& 7 & 10 & 10 & 1 \\
6 & 5 & 9 & 11 & 0 & 8 & 5 & 7 & 3 &
11 & 8 & 4 & 8 \\
3 & 7 & 5 & 3 & 8 & 5 & 3 & 4 & 11 &
5 & 6 & 11 & 6 \\
8 & 8 & 7 & 12 & 12 & 7 & 4 & 8 & 0
& 4 & 3 & 1 & 9 \\
10 & 8 & 10 & 4 & 12 & 3 & 11 & 0 & 4
& 8 & 3 & 10 & 7 \\
10 & 11 & 4 & 12 & 7 & 11 & 5 & 4 & 8
& 5 & 9 & 1 & 4 \\
10 & 7 & 8 & 11 & 10 & 8 & 6 & 3 & 3
& 9 & 11 & 0 & 8 \\
0 & 0 & 7 & 8 & 10 & 4 & 11 & 1 & 10
& 1 & 0 & 5 & 4 \\
9 & 4 & 11 & 11 & 1 & 8 & 6 & 9 & 7
& 4 & 8 & 4 & 10
\end{smatrix}$$


\item Symmetric self-dual [28,14,11] code over $GF(13)$ with $(I_{14} \mid A_{13}^{28,1})$ where \\
$$A_{13}^{28,1}= \begin{smatrix}
4 & 2 & 10 & 8 & 6 & 3 & 1 & 12 & 1 & 11 & 8 & 9 & 11 & 2 \\
2 & 6 & 2 & 10 & 5 & 8 & 12 & 10 & 1 & 11 & 6 & 12 & 1 & 8 \\
10 & 2 & 9 & 3 & 6 & 6 & 9 & 3 & 12 & 0 & 4 & 4 & 5 & 12 \\
8 & 10 & 3 & 8 & 12 & 4 & 7 & 7 & 5 & 1 & 1 & 3 & 11 & 7 \\
6 & 5 & 6 & 12 & 3 & 9 & 3 & 11 & 4 & 7 & 0 & 4 & 11 & 8 \\
3 & 8 & 6 & 4 & 9 & 11 & 9 & 12 & 8 & 7 & 1 & 0 & 5 & 6 \\
1 & 12 & 9 & 7 & 3 & 9 & 11 & 2 & 10 & 10 & 9 & 9 & 11 & 1 \\
12 & 10 & 3 & 7 & 11 & 12 & 2 & 6 & 1 & 4 & 7 & 5 & 4 & 0 \\
1 & 1 & 12 & 5 & 4 & 8 & 10 & 1 & 11 & 7 & 2 & 4 & 8 & 2 \\
11 & 11 & 0 & 1 & 7 & 7 & 10 & 4 & 7 & 3 & 12 & 1 & 9 & 8 \\
8 & 6 & 4 & 1 & 0 & 1 & 9 & 7 & 2 & 12 & 2 & 4 & 5 & 0 \\
9 & 12 & 4 & 3 & 4 & 0 & 9 & 5 & 4 & 1 & 4 & 7 & 11 & 10 \\
11 & 1 & 5 & 11 & 11 & 5 & 11 & 4 & 8 & 9 & 5 & 11 & 4 & 5 \\
2 & 8 & 12 & 7 & 8 & 6 & 1 & 0 & 2 & 8 & 0 & 10 & 5 & 9
\end{smatrix}.$$


\item Symmetric self-dual [30,15,11] code over $GF(13)$ with $(I_{15} \mid A_{13}^{30,1})$ where \\
$$A_{13}^{30,1}= \begin{smatrix}
11 & 10 & 8 & 9 & 2 & 1 & 4 & 12 & 12 & 7 & 12 & 2 & 2 & 6 & 6 \\
10 & 7 & 7 & 1 & 6 & 5 & 12 & 2 & 0 & 7 & 12 & 6 & 7 & 5 & 9 \\
8 & 7 & 10 & 0 & 11 & 12 & 10 & 5 & 3 & 11 & 4 & 7 & 0 & 4 & 11 \\
9 & 1 & 0 & 10 & 9 & 9 & 5 & 6 & 0 & 7 & 10 & 10 & 10 & 10 & 4 \\
2 & 6 & 11 & 9 & 5 & 4 & 11 & 2 & 2 & 1 & 9 & 11 & 0 & 2 & 11 \\
1 & 5 & 12 & 9 & 4 & 12 & 6 & 7 & 2 & 2 & 11 & 5 & 9 & 0 & 10 \\
4 & 12 & 10 & 5 & 11 & 6 & 12 & 12 & 2 & 0 & 10 & 8 & 7 & 0 & 1 \\
12 & 2 & 5 & 6 & 2 & 7 & 12 & 7 & 11 & 12 & 6 & 4 & 4 & 9 & 12 \\
12 & 0 & 3 & 0 & 2 & 2 & 2 & 11 & 2 & 3 & 0 & 2 & 0 & 2 & 11 \\
7 & 7 & 11 & 7 & 1 & 2 & 0 & 12 & 3 & 10 & 9 & 11 & 0 & 9 & 3 \\
12 & 12 & 4 & 10 & 9 & 11 & 10 & 6 & 0 & 9 & 12 & 7 & 9 & 7 & 6 \\
2 & 6 & 7 & 10 & 11 & 5 & 8 & 4 & 2 & 11 & 7 & 12 & 1 & 9 & 4 \\
2 & 7 & 0 & 10 & 0 & 9 & 7 & 4 & 0 & 0 & 9 & 1 & 4 & 2 & 1 \\
6 & 5 & 4 & 10 & 2 & 0 & 0 & 9 & 2 & 9 & 7 & 9 & 2 & 3 & 4 \\
6 & 9 & 11 & 4 & 11 & 10 & 1 & 12 & 11 & 3 & 6 & 4 & 1 & 4 & 8
\end{smatrix}$$

\item Symmetric self-dual [32,16,12] code  over $GF(13)$ with $(I_{16} \mid A_{13}^{32,1})$ where \\
$$A_{13}^{32,1}= \begin{smatrix}
11 & 5 & 8 & 5 & 2 & 7 & 11 & 11 & 10 & 12 & 2 & 11 & 12 & 3 & 4 & 7 \\
5 & 2 & 6 & 12 & 8 & 5 & 2 & 5 & 7 & 6 & 6 & 0 & 9 & 7 & 4 & 9 \\
8 & 6 & 11 & 3 & 2 & 3 & 4 & 11 & 7 & 6 & 8 & 11 & 12 & 2 & 7 & 6 \\
5 & 12 & 3 & 1 & 12 & 1 & 0 & 11 & 0 & 10 & 10 & 5 & 1 & 5 & 2 & 1 \\
2 & 8 & 2 & 12 & 7 & 5 & 12 & 8 & 4 & 8 & 4 & 0 & 5 & 12 & 4 & 0 \\
7 & 5 & 3 & 1 & 5 & 4 & 8 & 2 & 8 & 4 & 10 & 0 & 0 & 7 & 7 & 10 \\
11 & 2 & 4 & 0 & 12 & 8 & 9 & 3 & 9 & 7 & 5 & 8 & 10 & 7 & 6 & 1 \\
11 & 5 & 11 & 11 & 8 & 2 & 3 & 9 & 1 & 7 & 3 & 7 & 0 & 5 & 6 & 5 \\
10 & 7 & 7 & 0 & 4 & 8 & 9 & 1 & 10 & 12 & 10 & 8 & 5 & 1 & 5 & 5 \\
12 & 6 & 6 & 10 & 8 & 4 & 7 & 7 & 12 & 11 & 11 & 5 & 11 & 12 & 5 & 0 \\
2 & 6 & 8 & 10 & 4 & 10 & 5 & 3 & 10 & 11 & 7 & 12 & 6 & 2 & 3 & 12 \\
11 & 0 & 11 & 5 & 0 & 0 & 8 & 7 & 8 & 5 & 12 & 9 & 12 & 7 & 0 & 10 \\
12 & 9 & 12 & 1 & 5 & 0 & 10 & 0 & 5 & 11 & 6 & 12 & 8 & 0 & 12 & 6 \\
3 & 7 & 2 & 5 & 12 & 7 & 7 & 5 & 1 & 12 & 2 & 7 & 0 & 7 & 6 & 8 \\
4 & 4 & 7 & 2 & 4 & 7 & 6 & 6 & 5 & 5 & 3 & 0 & 12 & 6 & 4 & 9 \\
7 & 9 & 6 & 1 & 0 & 10 & 1 & 5 & 5 & 0 & 12 & 10 & 6 & 8 & 9 & 7
\end{smatrix}$$

\item Symmetric self-dual [34,17,12] code  over $GF(13)$ with $(I_{17} \mid A_{13}^{34,1})$ where \\
$$A_{13}^{34,1}= \begin{smatrix}
1 & 0 & 3 & 7 & 5 & 1 & 10 & 11 & 3 & 7 & 2 & 10 & 12 & 2 & 6 & 12 & 10 \\
0 & 11 & 5 & 8 & 5 & 2 & 7 & 11 & 11 & 10 & 12 & 2 & 11 & 12 & 3 & 4 & 7 \\
3 & 5 & 3 & 4 & 5 & 4 & 4 & 10 & 6 & 5 & 11 & 5 & 4 & 1 & 9 & 8 & 8 \\
7 & 8 & 4 & 2 & 4 & 10 & 5 & 1 & 9 & 11 & 9 & 10 & 3 & 2 & 11 & 12 & 8 \\
5 & 5 & 5 & 4 & 11 & 1 & 8 & 9 & 4 & 1 & 1 & 4 & 3 & 5 & 4 & 0 & 8 \\
1 & 2 & 4 & 10 & 1 & 10 & 9 & 6 & 4 & 12 & 1 & 8 & 10 & 11 & 4 & 1 & 4 \\
10 & 7 & 4 & 5 & 8 & 9 & 5 & 0 & 1 & 10 & 12 & 11 & 9 & 8 & 5 & 3 & 11 \\
11 & 11 & 10 & 1 & 9 & 6 & 0 & 8 & 11 & 6 & 8 & 10 & 1 & 11 & 10 & 12 & 6 \\
3 & 11 & 6 & 9 & 4 & 4 & 1 & 11 & 10 & 12 & 12 & 2 & 11 & 5 & 7 & 10 & 4 \\
7 & 10 & 5 & 11 & 1 & 12 & 10 & 6 & 12 & 1 & 2 & 12 & 0 & 8 & 10 & 10 & 7 \\
2 & 12 & 11 & 9 & 1 & 1 & 12 & 8 & 12 & 2 & 10 & 6 & 12 & 10 & 9 & 12 & 8 \\
10 & 2 & 5 & 10 & 4 & 8 & 11 & 10 & 2 & 12 & 6 & 8 & 8 & 1 & 0 & 12 & 0 \\
12 & 11 & 4 & 3 & 3 & 10 & 9 & 1 & 11 & 0 & 12 & 8 & 12 & 6 & 2 & 3 & 6 \\
2 & 12 & 1 & 2 & 5 & 11 & 8 & 11 & 5 & 8 & 10 & 1 & 6 & 7 & 10 & 6 & 1 \\
6 & 3 & 9 & 11 & 4 & 4 & 5 & 10 & 7 & 10 & 9 & 0 & 2 & 10 & 11 & 1 & 6 \\
12 & 4 & 8 & 12 & 0 & 1 & 3 & 12 & 10 & 10 & 12 & 12 & 3 & 6 & 1 & 7 & 5 \\
10 & 7 & 8 & 8 & 8 & 4 & 11 & 6 & 4 & 7 & 8 & 0 & 6 & 1 & 6 & 5 & 8
\end{smatrix}$$

\item Symmetric self-dual [36,18,13] code  over $GF(13)$ with $(I_{18} \mid A_{13}^{36,1})$ where \\
$$A_{13}^{36,1}= \begin{smatrix}
6 & 3 & 1 & 1 & 5 & 8 & 1 & 6 & 3 & 1 & 4 & 1 & 1 & 3 & 11 & 8 & 2 & 4 \\
3 & 3 & 5 & 8 & 6 & 6 & 6 & 1 & 0 & 8 & 1 & 7 & 2 & 1 & 5 & 7 & 9 & 4 \\
1 & 5 & 4 & 11 & 12 & 1 & 8 & 4 & 3 & 4 & 8 & 5 & 8 & 3 & 0 & 12 & 3 & 5 \\
1 & 8 & 11 & 9 & 8 & 1 & 10 & 1 & 2 & 12 & 3 & 4 & 11 & 9 & 2 & 5 & 7 & 6 \\
5 & 6 & 12 & 8 & 9 & 10 & 1 & 3 & 0 & 0 & 1 & 0 & 1 & 2 & 7 & 4 & 7 & 11 \\
8 & 6 & 1 & 1 & 10 & 5 & 10 & 10 & 10 & 0 & 11 & 10 & 0 & 4 & 0 & 11 & 5 & 5 \\
1 & 6 & 8 & 10 & 1 & 10 & 3 & 6 & 11 & 10 & 10 & 7 & 1 & 2 & 12 & 0 & 0 & 2 \\
6 & 1 & 4 & 1 & 3 & 10 & 6 & 0 & 4 & 11 & 11 & 9 & 8 & 0 & 1 & 7 & 10 & 12 \\
3 & 0 & 3 & 2 & 0 & 10 & 11 & 4 & 10 & 3 & 0 & 0 & 2 & 3 & 1 & 11 & 9 & 0 \\
1 & 8 & 4 & 12 & 0 & 0 & 10 & 11 & 3 & 3 & 10 & 5 & 8 & 3 & 6 & 3 & 9 & 2 \\
4 & 1 & 8 & 3 & 1 & 11 & 10 & 11 & 0 & 10 & 6 & 0 & 10 & 7 & 12 & 7 & 6 & 12 \\
1 & 7 & 5 & 4 & 0 & 10 & 7 & 9 & 0 & 5 & 0 & 3 & 12 & 4 & 11 & 5 & 11 & 6 \\
1 & 2 & 8 & 11 & 1 & 0 & 1 & 8 & 2 & 8 & 10 & 12 & 1 & 0 & 2 & 9 & 11 & 11 \\
3 & 1 & 3 & 9 & 2 & 4 & 2 & 0 & 3 & 3 & 7 & 4 & 0 & 1 & 9 & 3 & 0 & 0 \\
11 & 5 & 0 & 2 & 7 & 0 & 12 & 1 & 1 & 6 & 12 & 11 & 2 & 9 & 5 & 5 & 8 & 5 \\
8 & 7 & 12 & 5 & 4 & 11 & 0 & 7 & 11 & 3 & 7 & 5 & 9 & 3 & 5 & 5 & 6 & 3 \\
2 & 9 & 3 & 7 & 7 & 5 & 0 & 10 & 9 & 9 & 6 & 11 & 11 & 0 & 8 & 6 & 5 & 1 \\
4 & 4 & 5 & 6 & 11 & 5 & 2 & 12 & 0 & 2 & 12 & 6 & 11 & 0 & 5 & 3 & 1 & 0
\end{smatrix}$$

\item Symmetric self-dual [38,19,13] code over $GF(13)$ with $(I_{19} \mid A_{13}^{38,1})$ where\\ 
$$A_{13}^{38,1}= \begin{smatrix}
3 & 8 & 0 & 3 & 2 & 11 & 6 & 8 & 3 & 9 & 3 & 7 & 1 & 7 & 2 & 8 & 11 & 9 & 2 \\
8 & 0 & 3 & 2 & 6 & 0 & 10 & 8 & 7 & 6 & 2 & 2 & 10 & 12 & 8 & 5 & 3 & 5 & 9 \\
0 & 3 & 3 & 5 & 8 & 6 & 6 & 6 & 1 & 0 & 8 & 1 & 7 & 2 & 1 & 5 & 7 & 9 & 4 \\
3 & 2 & 5 & 6 & 8 & 2 & 5 & 9 & 6 & 9 & 6 & 4 & 10 & 4 & 0 & 1 & 2 & 9 & 2 \\
2 & 6 & 8 & 8 & 7 & 10 & 8 & 2 & 11 & 6 & 9 & 9 & 3 & 4 & 7 & 7 & 7 & 11 & 4 \\
11 & 0 & 6 & 2 & 10 & 7 & 3 & 9 & 6 & 9 & 3 & 8 & 1 & 8 & 4 & 2 & 2 & 3 & 0 \\
6 & 10 & 6 & 5 & 8 & 3 & 0 & 12 & 1 & 9 & 4 & 3 & 7 & 5 & 11 & 2 & 4 & 4 & 12 \\
8 & 8 & 6 & 9 & 2 & 9 & 12 & 10 & 7 & 1 & 11 & 8 & 3 & 12 & 7 & 6 & 8 & 3 & 7 \\
3 & 7 & 1 & 6 & 11 & 6 & 1 & 7 & 2 & 10 & 0 & 7 & 1 & 4 & 10 & 2 & 10 & 3 & 9 \\
9 & 6 & 0 & 9 & 6 & 9 & 9 & 1 & 10 & 2 & 9 & 1 & 2 & 3 & 7 & 4 & 7 & 1 & 4 \\
3 & 2 & 8 & 6 & 9 & 3 & 4 & 11 & 0 & 9 & 5 & 6 & 10 & 4 & 0 & 7 & 6 & 2 & 12 \\
7 & 2 & 1 & 4 & 9 & 8 & 3 & 8 & 7 & 1 & 6 & 1 & 3 & 5 & 0 & 10 & 1 & 7 & 5 \\
1 & 10 & 7 & 10 & 3 & 1 & 7 & 3 & 1 & 2 & 10 & 3 & 9 & 2 & 3 & 7 & 6 & 0 & 5 \\
7 & 12 & 2 & 4 & 4 & 8 & 5 & 12 & 4 & 3 & 4 & 5 & 2 & 9 & 6 & 0 & 3 & 12 & 4 \\
2 & 8 & 1 & 0 & 7 & 4 & 11 & 7 & 10 & 7 & 0 & 0 & 3 & 6 & 12 & 1 & 5 & 4 & 11 \\
8 & 5 & 5 & 1 & 7 & 2 & 2 & 6 & 2 & 4 & 7 & 10 & 7 & 0 & 1 & 12 & 0 & 11 & 10 \\
11 & 3 & 7 & 2 & 7 & 2 & 4 & 8 & 10 & 7 & 6 & 1 & 6 & 3 & 5 & 0 & 3 & 2 & 5 \\
9 & 5 & 9 & 9 & 11 & 3 & 4 & 3 & 3 & 1 & 2 & 7 & 0 & 12 & 4 & 11 & 2 & 10 & 5 \\
2 & 9 & 4 & 2 & 4 & 0 & 12 & 7 & 9 & 4 & 12 & 5 & 5 & 4 & 11 & 10 & 5 & 5 & 11
\end{smatrix}$$

\item Symmetric self-dual [40,20,14] code over $GF(13)$ with $(I_{20} \mid A_{13}^{40,1})$ where \\
$$A_{13}^{40,1}= \begin{smatrix}
8 & 5 & 10 & 5 & 4 & 1 & 8 & 1 & 2 & 3 & 4 & 11 & 5 & 8 & 6 & 3 & 2 & 12 & 9 & 3 \\
5 & 7 & 3 & 4 & 1 & 8 & 7 & 12 & 7 & 8 & 7 & 4 & 11 & 10 & 4 & 7 & 7 & 5 & 11 & 7 \\
10 & 3 & 3 & 11 & 11 & 5 & 5 & 9 & 6 & 4 & 2 & 4 & 10 & 2 & 6 & 5 & 3 & 4 & 9 & 6 \\
5 & 4 & 11 & 7 & 3 & 1 & 2 & 12 & 5 & 6 & 11 & 9 & 5 & 3 & 12 & 6 & 4 & 1 & 11 & 9 \\
4 & 1 & 11 & 3 & 7 & 5 & 4 & 2 & 3 & 10 & 10 & 12 & 2 & 12 & 12 & 4 & 8 & 5 & 8 & 6 \\
1 & 8 & 5 & 1 & 5 & 3 & 4 & 4 & 7 & 12 & 3 & 4 & 2 & 10 & 6 & 8 & 12 & 11 & 1 & 5 \\
8 & 7 & 5 & 2 & 4 & 4 & 11 & 10 & 10 & 1 & 11 & 2 & 4 & 5 & 11 & 12 & 3 & 8 & 1 & 8 \\
1 & 12 & 9 & 12 & 2 & 4 & 10 & 9 & 4 & 2 & 6 & 12 & 9 & 1 & 7 & 12 & 7 & 8 & 7 & 0 \\
2 & 7 & 6 & 5 & 3 & 7 & 10 & 4 & 7 & 9 & 8 & 1 & 7 & 4 & 3 & 9 & 3 & 3 & 9 & 9 \\
3 & 8 & 4 & 6 & 10 & 12 & 1 & 2 & 9 & 5 & 1 & 11 & 12 & 9 & 10 & 0 & 4 & 9 & 12 & 12 \\
4 & 7 & 2 & 11 & 10 & 3 & 11 & 6 & 8 & 1 & 3 & 2 & 12 & 4 & 11 & 11 & 11 & 10 & 0 & 8 \\
11 & 4 & 4 & 9 & 12 & 4 & 2 & 12 & 1 & 11 & 2 & 2 & 7 & 9 & 0 & 11 & 10 & 11 & 9 & 10 \\
5 & 11 & 10 & 5 & 2 & 2 & 4 & 9 & 7 & 12 & 12 & 7 & 5 & 12 & 2 & 5 & 9 & 8 & 9 & 10 \\
8 & 10 & 2 & 3 & 12 & 10 & 5 & 1 & 4 & 9 & 4 & 9 & 12 & 0 & 5 & 11 & 8 & 12 & 11 & 0 \\
6 & 4 & 6 & 12 & 12 & 6 & 11 & 7 & 3 & 10 & 11 & 0 & 2 & 5 & 8 & 12 & 4 & 1 & 4 & 10 \\
3 & 7 & 5 & 6 & 4 & 8 & 12 & 12 & 9 & 0 & 11 & 11 & 5 & 11 & 12 & 2 & 3 & 4 & 0 & 1 \\
2 & 7 & 3 & 4 & 8 & 12 & 3 & 7 & 3 & 4 & 11 & 10 & 9 & 8 & 4 & 3 & 9 & 8 & 4 & 12 \\
12 & 5 & 4 & 1 & 5 & 11 & 8 & 8 & 3 & 9 & 10 & 11 & 8 & 12 & 1 & 4 & 8 & 12 & 12 & 4 \\
9 & 11 & 9 & 11 & 8 & 1 & 1 & 7 & 9 & 12 & 0 & 9 & 9 & 11 & 4 & 0 & 4 & 12 & 11 & 1 \\
3 & 7 & 6 & 9 & 6 & 5 & 8 & 0 & 9 & 12 & 8 & 10 & 10 & 0 & 10 & 1 & 12 & 4 & 1 & 1
\end{smatrix}$$
\end{itemize}

\section{\\Self-dual codes over $GF(17)$}\label{GF(17)}

\begin{itemize}
\item Symmetric self-dual [28,14,10] code over $GF(17)$ with $(I_{14} \mid A_{17}^{28,1})$ where \\
$$A_{17}^{28,1}= \begin{smatrix}
4 & 2 & 4 & 9 & 9 & 7 & 16 & 7 & 13 & 4 & 14 & 11 & 1 & 7 \\
2 & 14 & 16 & 14 & 12 & 3 & 1 & 0 & 3 & 0 & 5 & 3 & 4 & 16 \\
4 & 16 & 4 & 2 & 0 & 5 & 16 & 13 & 2 & 3 & 12 & 9 & 16 & 2 \\
9 & 14 & 2 & 16 & 12 & 0 & 15 & 14 & 8 & 16 & 7 & 14 & 11 & 9 \\
9 & 12 & 0 & 12 & 12 & 7 & 0 & 4 & 13 & 2 & 10 & 1 & 9 & 1 \\
7 & 3 & 5 & 0 & 7 & 13 & 12 & 5 & 2 & 7 & 14 & 5 & 2 & 13 \\
16 & 1 & 16 & 15 & 0 & 12 & 4 & 14 & 11 & 8 & 9 & 8 & 11 & 1 \\
7 & 0 & 13 & 14 & 4 & 5 & 14 & 13 & 8 & 11 & 5 & 8 & 16 & 3 \\
13 & 3 & 2 & 8 & 13 & 2 & 11 & 8 & 14 & 9 & 12 & 9 & 9 & 6 \\
4 & 0 & 3 & 16 & 2 & 7 & 8 & 11 & 9 & 3 & 1 & 16 & 10 & 11 \\
14 & 5 & 12 & 7 & 10 & 14 & 9 & 5 & 12 & 1 & 7 & 4 & 14 & 1 \\
11 & 3 & 9 & 14 & 1 & 5 & 8 & 8 & 9 & 16 & 4 & 6 & 11 & 4 \\
1 & 4 & 16 & 11 & 9 & 2 & 11 & 16 & 9 & 10 & 14 & 11 & 15 & 1 \\
7 & 16 & 2 & 9 & 1 & 13 & 1 & 3 & 6 & 11 & 1 & 4 & 1 & 11
\end{smatrix}$$,


\item Symmetric self-dual [30,15,12] code over $GF(17)$ with $(I_{15} \mid A_{17}^{30,1})$ where \\
$$A_{17}^{30,1}= \begin{smatrix}
14 & 14 & 14 & 0 & 0 & 15 & 9 & 9 & 8 & 1 & 12 & 1 & 2 & 8 & 15 \\
14 & 13 & 11 & 4 & 9 & 15 & 14 & 6 & 0 & 10 & 2 & 11 & 5 & 11 & 13 \\
14 & 11 & 6 & 16 & 14 & 1 & 10 & 8 & 10 & 0 & 15 & 2 & 14 & 14 & 5 \\
0 & 4 & 16 & 4 & 2 & 0 & 5 & 16 & 13 & 2 & 3 & 12 & 9 & 16 & 2 \\
0 & 9 & 14 & 2 & 16 & 12 & 0 & 15 & 14 & 8 & 16 & 7 & 14 & 11 & 9 \\
15 & 15 & 1 & 0 & 12 & 16 & 6 & 16 & 5 & 11 & 12 & 8 & 14 & 10 & 5 \\
9 & 14 & 10 & 5 & 0 & 6 & 9 & 8 & 9 & 11 & 13 & 6 & 6 & 6 & 12 \\
9 & 6 & 8 & 16 & 15 & 16 & 8 & 0 & 1 & 3 & 14 & 1 & 9 & 15 & 0 \\
8 & 0 & 10 & 13 & 14 & 5 & 9 & 1 & 9 & 16 & 5 & 13 & 7 & 12 & 4 \\
1 & 10 & 0 & 2 & 8 & 11 & 11 & 3 & 16 & 15 & 4 & 13 & 11 & 0 & 4 \\
12 & 2 & 15 & 3 & 16 & 12 & 13 & 14 & 5 & 4 & 11 & 13 & 6 & 4 & 4 \\
1 & 11 & 2 & 12 & 7 & 8 & 6 & 1 & 13 & 13 & 13 & 8 & 6 & 5 & 16 \\
2 & 5 & 14 & 9 & 14 & 14 & 6 & 9 & 7 & 11 & 6 & 6 & 10 & 10 & 0 \\
8 & 11 & 14 & 16 & 11 & 10 & 6 & 15 & 12 & 0 & 4 & 5 & 10 & 11 & 2 \\
15 & 13 & 5 & 2 & 9 & 5 & 12 & 0 & 4 & 4 & 4 & 16 & 0 & 2 & 15
\end{smatrix}$$

\item Symmetric self-dual [32,16,12] code over $GF(17)$ with $(I_{16} \mid A_{17}^{32,1})$ where \\
$$A_{17}^{32,1}= \begin{smatrix}
11 & 9 & 4 & 10 & 11 & 6 & 4 & 0 & 9 & 7 & 7 & 14 & 4 & 15 & 13 & 7 \\
9 & 11 & 7 & 5 & 2 & 15 & 8 & 9 & 6 & 0 & 10 & 13 & 11 & 14 & 15 & 7 \\
4 & 7 & 8 & 7 & 3 & 10 & 10 & 14 & 16 & 4 & 14 & 10 & 6 & 16 & 16 & 0 \\
10 & 5 & 7 & 13 & 5 & 8 & 14 & 10 & 16 & 3 & 10 & 1 & 15 & 16 & 1 & 15 \\
11 & 2 & 3 & 5 & 14 & 9 & 16 & 5 & 1 & 7 & 13 & 8 & 11 & 1 & 0 & 13 \\
6 & 15 & 10 & 8 & 9 & 9 & 13 & 0 & 13 & 3 & 14 & 11 & 8 & 5 & 10 & 15 \\
4 & 8 & 10 & 14 & 16 & 13 & 11 & 6 & 9 & 9 & 15 & 3 & 3 & 8 & 15 & 9 \\
0 & 9 & 14 & 10 & 5 & 0 & 6 & 9 & 8 & 9 & 11 & 13 & 6 & 6 & 6 & 12 \\
9 & 6 & 16 & 16 & 1 & 13 & 9 & 8 & 14 & 10 & 12 & 15 & 11 & 4 & 5 & 9 \\
7 & 0 & 4 & 3 & 7 & 3 & 9 & 9 & 10 & 16 & 6 & 2 & 0 & 5 & 8 & 11 \\
7 & 10 & 14 & 10 & 13 & 14 & 15 & 11 & 12 & 6 & 5 & 1 & 0 & 9 & 13 & 11 \\
14 & 13 & 10 & 1 & 8 & 11 & 3 & 13 & 15 & 2 & 1 & 5 & 4 & 2 & 13 & 1 \\
4 & 11 & 6 & 15 & 11 & 8 & 3 & 6 & 11 & 0 & 0 & 4 & 3 & 0 & 10 & 3 \\
15 & 14 & 16 & 16 & 1 & 5 & 8 & 6 & 4 & 5 & 9 & 2 & 0 & 13 & 16 & 15 \\
13 & 15 & 16 & 1 & 0 & 10 & 15 & 6 & 5 & 8 & 13 & 13 & 10 & 16 & 6 & 15 \\
7 & 7 & 0 & 15 & 13 & 15 & 9 & 12 & 9 & 11 & 11 & 1 & 3 & 15 & 15 & 5
\end{smatrix}$$

\item Symmetric self-dual [34,17,12] code over $GF(17)$ with $(I_{17} \mid A_{17}^{34,1})$ where \\
$$A_{17}^{34,1}= \begin{smatrix}
1 & 3 & 16 & 5 & 0 & 0 & 0 & 11 & 7 & 7 & 0 & 6 & 6 & 5 & 7 & 2 & 11 \\
3 & 8 & 10 & 16 & 10 & 11 & 6 & 10 & 10 & 2 & 7 & 1 & 8 & 16 & 8 & 11 & 13 \\
16 & 10 & 5 & 3 & 5 & 2 & 15 & 6 & 0 & 14 & 0 & 12 & 15 & 7 & 5 & 10 & 5 \\
5 & 16 & 3 & 11 & 7 & 3 & 10 & 3 & 8 & 10 & 4 & 4 & 0 & 9 & 10 & 7 & 10 \\
0 & 10 & 5 & 7 & 13 & 5 & 8 & 14 & 10 & 16 & 3 & 10 & 1 & 15 & 16 & 1 & 15 \\
0 & 11 & 2 & 3 & 5 & 14 & 9 & 16 & 5 & 1 & 7 & 13 & 8 & 11 & 1 & 0 & 13 \\
0 & 6 & 15 & 10 & 8 & 9 & 9 & 13 & 0 & 13 & 3 & 14 & 11 & 8 & 5 & 10 & 15 \\
11 & 10 & 6 & 3 & 14 & 16 & 13 & 16 & 3 & 6 & 9 & 10 & 15 & 13 & 5 & 2 & 14 \\
7 & 10 & 0 & 8 & 10 & 5 & 0 & 3 & 4 & 3 & 9 & 14 & 16 & 0 & 1 & 7 & 9 \\
7 & 2 & 14 & 10 & 16 & 1 & 13 & 6 & 3 & 9 & 10 & 15 & 1 & 5 & 16 & 6 & 6 \\
0 & 7 & 0 & 4 & 3 & 7 & 3 & 9 & 9 & 10 & 16 & 6 & 2 & 0 & 5 & 8 & 11 \\
6 & 1 & 12 & 4 & 10 & 13 & 14 & 10 & 14 & 15 & 6 & 10 & 6 & 7 & 12 & 9 & 6 \\
6 & 8 & 15 & 0 & 1 & 8 & 11 & 15 & 16 & 1 & 2 & 6 & 10 & 11 & 5 & 9 & 13 \\
5 & 16 & 7 & 9 & 15 & 11 & 8 & 13 & 0 & 5 & 0 & 7 & 11 & 6 & 11 & 1 & 13 \\
7 & 8 & 5 & 10 & 16 & 1 & 5 & 5 & 1 & 16 & 5 & 12 & 5 & 11 & 8 & 0 & 12 \\
2 & 11 & 10 & 7 & 1 & 0 & 10 & 2 & 7 & 6 & 8 & 9 & 9 & 1 & 0 & 16 & 2 \\
11 & 13 & 5 & 10 & 15 & 13 & 15 & 14 & 9 & 6 & 11 & 6 & 13 & 13 & 12 & 2 & 10
\end{smatrix}$$

\item Symmetric self-dual [36,18,13] code over $GF(17)$ with $(I_{18} \mid A_{17}^{36,1})$ where \\
$$A_{17}^{36,1}= \begin{smatrix}
7 & 10 & 4 & 7 & 7 & 6 & 14 & 9 & 5 & 6 & 9 & 8 & 14 & 13 & 7 & 4 & 6 & 14 \\
10 & 6 & 5 & 11 & 0 & 3 & 7 & 13 & 5 & 10 & 3 & 4 & 13 & 4 & 0 & 9 & 5 & 1 \\
4 & 5 & 2 & 8 & 14 & 1 & 7 & 1 & 11 & 1 & 14 & 12 & 14 & 14 & 14 & 2 & 2 & 9 \\
7 & 11 & 8 & 10 & 8 & 2 & 12 & 2 & 12 & 14 & 1 & 13 & 5 & 0 & 12 & 3 & 7 & 15 \\
7 & 0 & 14 & 8 & 16 & 4 & 13 & 14 & 9 & 5 & 14 & 0 & 14 & 2 & 14 & 8 & 4 & 3 \\
6 & 3 & 1 & 2 & 4 & 8 & 16 & 9 & 7 & 5 & 0 & 2 & 4 & 10 & 12 & 7 & 13 & 9 \\
14 & 7 & 7 & 12 & 13 & 16 & 0 & 0 & 11 & 16 & 9 & 16 & 16 & 12 & 4 & 14 & 11 & 16 \\
9 & 13 & 1 & 2 & 14 & 9 & 0 & 2 & 11 & 1 & 6 & 10 & 5 & 16 & 12 & 0 & 11 & 6 \\
5 & 5 & 11 & 12 & 9 & 7 & 11 & 11 & 13 & 13 & 4 & 11 & 5 & 14 & 2 & 6 & 12 & 9 \\
6 & 10 & 1 & 14 & 5 & 5 & 16 & 1 & 13 & 16 & 4 & 8 & 8 & 8 & 14 & 9 & 2 & 3 \\
9 & 3 & 14 & 1 & 14 & 0 & 9 & 6 & 4 & 4 & 2 & 0 & 6 & 6 & 9 & 11 & 7 & 14 \\
8 & 4 & 12 & 13 & 0 & 2 & 16 & 10 & 11 & 8 & 0 & 9 & 15 & 14 & 13 & 10 & 7 & 3 \\
14 & 13 & 14 & 5 & 14 & 4 & 16 & 5 & 5 & 8 & 6 & 15 & 13 & 10 & 0 & 8 & 3 & 9 \\
13 & 4 & 14 & 0 & 2 & 10 & 12 & 16 & 14 & 8 & 6 & 14 & 10 & 4 & 13 & 11 & 1 & 0 \\
7 & 0 & 14 & 12 & 14 & 12 & 4 & 12 & 2 & 14 & 9 & 13 & 0 & 13 & 11 & 9 & 15 & 6 \\
4 & 9 & 2 & 3 & 8 & 7 & 14 & 0 & 6 & 9 & 11 & 10 & 8 & 11 & 9 & 2 & 8 & 8 \\
6 & 5 & 2 & 7 & 4 & 13 & 11 & 11 & 12 & 2 & 7 & 7 & 3 & 1 & 15 & 8 & 11 & 13 \\
14 & 1 & 9 & 15 & 3 & 9 & 16 & 6 & 9 & 3 & 14 & 3 & 9 & 0 & 6 & 8 & 13 & 13
\end{smatrix}$$

\item Symmetric self-dual [38,19,14] code over $GF(17)$ with $(I_{19} \mid A_{17}^{38,1})$ where \\
$$A_{17}^{38,1}= \begin{smatrix}
4 & 1 & 9 & 8 & 8 & 10 & 7 & 13 & 1 & 9 & 1 & 10 & 9 & 0 & 10 & 16 & 5 & 2 & 9 \\
1 & 5 & 9 & 5 & 8 & 4 & 9 & 5 & 7 & 4 & 4 & 6 & 7 & 14 & 10 & 9 & 11 & 2 & 13 \\
9 & 9 & 14 & 14 & 3 & 7 & 13 & 11 & 12 & 13 & 9 & 10 & 12 & 13 & 11 & 1 & 4 & 3 & 9 \\
8 & 5 & 14 & 10 & 16 & 7 & 8 & 3 & 2 & 3 & 2 & 7 & 4 & 14 & 7 & 13 & 7 & 4 & 1 \\
8 & 8 & 3 & 16 & 1 & 1 & 9 & 8 & 3 & 4 & 15 & 11 & 5 & 5 & 10 & 11 & 8 & 9 & 7 \\
10 & 4 & 7 & 7 & 1 & 3 & 0 & 8 & 11 & 16 & 2 & 1 & 7 & 14 & 6 & 0 & 10 & 15 & 10 \\
7 & 9 & 13 & 8 & 9 & 0 & 12 & 4 & 12 & 0 & 8 & 13 & 12 & 4 & 6 & 9 & 5 & 2 & 2 \\
13 & 5 & 11 & 3 & 8 & 8 & 4 & 2 & 8 & 15 & 7 & 4 & 3 & 16 & 7 & 13 & 3 & 10 & 3 \\
1 & 7 & 12 & 2 & 3 & 11 & 12 & 8 & 0 & 10 & 16 & 3 & 9 & 5 & 13 & 14 & 7 & 7 & 5 \\
9 & 4 & 13 & 3 & 4 & 16 & 0 & 15 & 10 & 4 & 12 & 11 & 2 & 5 & 4 & 3 & 1 & 10 & 0 \\
1 & 4 & 9 & 2 & 15 & 2 & 8 & 7 & 16 & 12 & 14 & 1 & 7 & 8 & 5 & 16 & 16 & 15 & 2 \\
10 & 6 & 10 & 7 & 11 & 1 & 13 & 4 & 3 & 11 & 1 & 6 & 7 & 6 & 10 & 12 & 13 & 1 & 4 \\
9 & 7 & 12 & 4 & 5 & 7 & 12 & 3 & 9 & 2 & 7 & 7 & 0 & 15 & 4 & 14 & 5 & 5 & 11 \\
0 & 14 & 13 & 14 & 5 & 14 & 4 & 16 & 5 & 5 & 8 & 6 & 15 & 13 & 10 & 0 & 8 & 3 & 9 \\
10 & 10 & 11 & 7 & 10 & 6 & 6 & 7 & 13 & 4 & 5 & 10 & 4 & 10 & 8 & 16 & 13 & 12 & 7 \\
16 & 9 & 1 & 13 & 11 & 0 & 9 & 13 & 14 & 3 & 16 & 12 & 14 & 0 & 16 & 9 & 2 & 2 & 7 \\
5 & 11 & 4 & 7 & 8 & 10 & 5 & 3 & 7 & 1 & 16 & 13 & 5 & 8 & 13 & 2 & 3 & 5 & 3 \\
2 & 2 & 3 & 4 & 9 & 15 & 2 & 10 & 7 & 10 & 15 & 1 & 5 & 3 & 12 & 2 & 5 & 3 & 11 \\
9 & 13 & 9 & 1 & 7 & 10 & 2 & 3 & 5 & 0 & 2 & 4 & 11 & 9 & 7 & 7 & 3 & 11 & 4
\end{smatrix}$$

\item Symmetric self-dual [40,20,14] code over $GF(17)$ with $(I_{20} \mid A_{17}^{40,1})$ where \\
$$A_{17}^{40,1}= \begin{smatrix}
9 & 12 & 9 & 13 & 3 & 0 & 3 & 0 & 12 & 15 & 16 & 3 & 6 & 15 & 6 & 15 & 13 & 10 & 10 & 2 \\
12 & 9 & 9 & 13 & 5 & 8 & 7 & 7 & 1 & 3 & 10 & 15 & 4 & 11 & 11 & 12 & 3 & 12 & 9 & 7 \\
9 & 9 & 11 & 12 & 7 & 8 & 6 & 9 & 13 & 0 & 9 & 6 & 10 & 0 & 1 & 3 & 12 & 12 & 3 & 3 \\
13 & 13 & 12 & 7 & 15 & 3 & 8 & 13 & 15 & 0 & 7 & 10 & 12 & 0 & 15 & 16 & 11 & 13 & 12 & 4 \\
3 & 5 & 7 & 15 & 5 & 16 & 2 & 8 & 0 & 11 & 16 & 14 & 14 & 13 & 4 & 16 & 14 & 13 & 10 & 9 \\
0 & 8 & 8 & 3 & 16 & 1 & 1 & 9 & 8 & 3 & 4 & 15 & 11 & 5 & 5 & 10 & 11 & 8 & 9 & 7 \\
3 & 7 & 6 & 8 & 2 & 1 & 15 & 0 & 5 & 3 & 12 & 14 & 8 & 16 & 4 & 15 & 1 & 16 & 4 & 1 \\
0 & 7 & 9 & 13 & 8 & 9 & 0 & 12 & 4 & 12 & 0 & 8 & 13 & 12 & 4 & 6 & 9 & 5 & 2 & 2 \\
12 & 1 & 13 & 15 & 0 & 8 & 5 & 4 & 7 & 10 & 16 & 4 & 15 & 5 & 10 & 9 & 0 & 10 & 0 & 1 \\
15 & 3 & 0 & 0 & 11 & 3 & 3 & 12 & 10 & 11 & 7 & 8 & 4 & 3 & 6 & 7 & 2 & 3 & 3 & 11 \\
16 & 10 & 9 & 7 & 16 & 4 & 12 & 0 & 16 & 7 & 11 & 8 & 3 & 16 & 14 & 1 & 14 & 16 & 8 & 3 \\
3 & 15 & 6 & 10 & 14 & 15 & 14 & 8 & 4 & 8 & 8 & 9 & 8 & 16 & 15 & 14 & 0 & 5 & 4 & 10 \\
6 & 4 & 10 & 12 & 14 & 11 & 8 & 13 & 15 & 4 & 3 & 8 & 3 & 8 & 3 & 11 & 14 & 8 & 13 & 3 \\
15 & 11 & 0 & 0 & 13 & 5 & 16 & 12 & 5 & 3 & 16 & 16 & 8 & 11 & 16 & 15 & 2 & 1 & 1 & 0 \\
6 & 11 & 1 & 15 & 4 & 5 & 4 & 4 & 10 & 6 & 14 & 15 & 3 & 16 & 10 & 11 & 2 & 3 & 15 & 8 \\
15 & 12 & 3 & 16 & 16 & 10 & 15 & 6 & 9 & 7 & 1 & 14 & 11 & 15 & 11 & 2 & 4 & 9 & 8 & 13 \\
13 & 3 & 12 & 11 & 14 & 11 & 1 & 9 & 0 & 2 & 14 & 0 & 14 & 2 & 2 & 4 & 2 & 11 & 11 & 2 \\
10 & 12 & 12 & 13 & 13 & 8 & 16 & 5 & 10 & 3 & 16 & 5 & 8 & 1 & 3 & 9 & 11 & 6 & 8 & 7 \\
10 & 9 & 3 & 12 & 10 & 9 & 4 & 2 & 0 & 3 & 8 & 4 & 13 & 1 & 15 & 8 & 11 & 8 & 6 & 15 \\
2 & 7 & 3 & 4 & 9 & 7 & 1 & 2 & 1 & 11 & 3 & 10 & 3 & 0 & 8 & 13 & 2 & 7 & 15 & 15
\end{smatrix}$$

\end{itemize}

\section{\\Quadratic residue codes over various finite fields}\label{Quadratic residue codes of new parameters}

\begin{tabular}{ll}\label{QR codes}

$[32, 16, 14]$ code over $GF(19)$: & $[20, 10, 10]$ code over $GF(23)$: \\

$\begin{smatrix}
18 & 13 & 17 & 11 & 10 & 15 & 15 & 8 & 3 & 12 & 4 & 12 & 0 & 10 & 14 & 18 \\
14 & 7 & 3 & 15 & 4 & 9 & 14 & 17 & 4 & 6 & 13 & 7 & 12 & 12 & 4 & 13 \\
4 & 0 & 15 & 16 & 13 & 1 & 6 & 1 & 5 & 13 & 9 & 3 & 7 & 10 & 13 & 17 \\
13 & 6 & 7 & 5 & 0 & 8 & 15 & 16 & 0 & 1 & 18 & 5 & 3 & 10 & 18 & 11 \\
18 & 7 & 4 & 18 & 15 & 15 & 4 & 4 & 0 & 12 & 5 & 11 & 5 & 13 & 5 & 10 \\
5 & 10 & 17 & 6 & 6 & 16 & 16 & 2 & 8 & 16 & 11 & 2 & 11 & 12 & 0 & 15 \\
0 & 5 & 10 & 17 & 6 & 6 & 16 & 16 & 2 & 8 & 16 & 11 & 2 & 11 & 12 & 15 \\
12 & 15 & 10 & 11 & 11 & 16 & 16 & 15 & 18 & 10 & 17 & 5 & 11 & 15 & 14 & 8 \\
14 & 1 & 5 & 8 & 4 & 10 & 15 & 18 & 11 & 2 & 11 & 1 & 5 & 4 & 9 & 3 \\
9 & 11 & 0 & 1 & 13 & 2 & 8 & 0 & 10 & 17 & 4 & 17 & 1 & 10 & 11 & 12 \\
11 & 18 & 14 & 12 & 5 & 0 & 8 & 15 & 5 & 11 & 11 & 5 & 17 & 5 & 8 & 4 \\
8 & 2 & 15 & 2 & 8 & 18 & 13 & 1 & 10 & 4 & 17 & 10 & 5 & 13 & 7 & 12 \\
7 & 12 & 16 & 14 & 8 & 17 & 8 & 14 & 18 & 2 & 14 & 9 & 10 & 11 & 10 & 0 \\
10 & 10 & 13 & 1 & 9 & 10 & 0 & 4 & 3 & 12 & 0 & 8 & 9 & 5 & 4 & 10 \\
4 & 15 & 18 & 7 & 18 & 6 & 7 & 6 & 11 & 12 & 15 & 9 & 8 & 7 & 6 & 14 \\
6 & 2 & 8 & 9 & 4 & 4 & 11 & 16 & 7 & 15 & 7 & 0 & 9 & 5 & 18 & 1
\end{smatrix},$
&
$\begin{smatrix}
22 & 12 & 2 & 9 & 10 & 15 & 12 & 21 & 13 & 1 \\
13 & 4 & 9 & 0 & 17 & 22 & 20 & 15 & 13 & 11 \\
13 & 18 & 1 & 7 & 8 & 6 & 4 & 0 & 7 & 21 \\
7 & 21 & 4 & 7 & 6 & 18 & 14 & 18 & 1 & 14 \\
1 & 18 & 19 & 18 & 20 & 14 & 6 & 16 & 5 & 13 \\
5 & 10 & 8 & 20 & 14 & 14 & 0 & 16 & 20 & 8 \\
20 & 18 & 16 & 12 & 4 & 13 & 4 & 17 & 9 & 11 \\
9 & 4 & 0 & 4 & 14 & 7 & 20 & 22 & 15 & 2 \\
15 & 13 & 20 & 3 & 15 & 19 & 11 & 4 & 11 & 10 \\
11 & 21 & 14 & 13 & 8 & 11 & 2 & 10 & 22 & 22
\end{smatrix},$
\\&
\\   
$[24, 12, 12]$ code over $GF(29)$: &$[24, 12, 12]$ code over $GF(41)$: \\

$\begin{smatrix}
28 & 18 & 21 & 4 & 14 & 23 & 19 & 7 & 25 & 16 & 19 & 1 \\
19 & 5 & 25 & 3 & 28 & 12 & 10 & 2 & 25 & 11 & 3 & 11 \\
3 & 23 & 0 & 13 & 19 & 17 & 13 & 18 & 14 & 6 & 12 & 8 \\
12 & 19 & 3 & 10 & 19 & 4 & 21 & 16 & 8 & 25 & 10 & 25 \\
10 & 6 & 12 & 21 & 15 & 21 & 17 & 9 & 27 & 22 & 9 & 15 \\
9 & 22 & 20 & 5 & 11 & 11 & 24 & 12 & 16 & 28 & 25 & 6 \\
25 & 23 & 19 & 7 & 3 & 16 & 0 & 23 & 25 & 22 & 17 & 10 \\
17 & 9 & 14 & 9 & 1 & 18 & 12 & 26 & 4 & 14 & 18 & 22 \\
18 & 12 & 8 & 0 & 18 & 22 & 24 & 2 & 11 & 6 & 20 & 4 \\
20 & 6 & 27 & 15 & 10 & 22 & 19 & 0 & 24 & 10 & 3 & 13 \\
3 & 24 & 1 & 15 & 2 & 28 & 23 & 27 & 12 & 5 & 11 & 10 \\
11 & 8 & 25 & 15 & 6 & 10 & 22 & 4 & 13 & 10 & 28 & 28
\end{smatrix},$
&

$\begin{smatrix}
40 & 25 & 28 & 4 & 19 & 33 & 29 & 12 & 37 & 23 & 26 & 40 \\
26 & 5 & 35 & 6 & 2 & 22 & 17 & 4 & 34 & 13 & 3 & 25 \\
3 & 33 & 3 & 23 & 31 & 26 & 17 & 22 & 16 & 6 & 17 & 28 \\
17 & 29 & 8 & 17 & 28 & 3 & 25 & 18 & 8 & 35 & 15 & 4 \\
15 & 11 & 19 & 30 & 19 & 25 & 19 & 9 & 37 & 32 & 14 & 19 \\
14 & 34 & 29 & 4 & 10 & 8 & 29 & 15 & 24 & 2 & 37 & 33 \\
37 & 32 & 23 & 4 & 39 & 19 & 1 & 36 & 40 & 34 & 24 & 29 \\
24 & 11 & 16 & 9 & 40 & 26 & 20 & 0 & 9 & 21 & 25 & 12 \\
25 & 14 & 8 & 39 & 26 & 35 & 39 & 7 & 18 & 8 & 27 & 37 \\
27 & 6 & 37 & 23 & 18 & 37 & 31 & 2 & 33 & 12 & 3 & 23 \\
3 & 34 & 4 & 25 & 7 & 1 & 32 & 36 & 14 & 5 & 16 & 26 \\
16 & 13 & 37 & 22 & 8 & 12 & 29 & 4 & 18 & 15 & 40 & 1
\end{smatrix}.$
\\ 
&
\\   
%
%
%

\end{tabular}


\begin{thebibliography}{00}



\bibitem{arasu2001} K.T. ARASU and T.A. GULLIVER, ``Self-dual codes over $\F_p$ and weighing matrices," {\it IEEE Trans. Inform. Theory}, 47(5), 2051-2055, 2001.

\bibitem{Bannai1999} E. Bannai, S. T. Dougherty, M. Harada, and M. Oura, ``Type II codes, even unimodular lattices, and invariant rings,'' {\it IEEE Transactions on Information Theory}, 45(4): 1194-1205, 1999.

\bibitem{BeEry2019} I. Be’Ery, N. Raviv, T. Raviv, and Y. Be’Ery, ``Active deep decoding of linear codes,'' {\it IEEE Transactions on Communications}, 68(2), 728-736, 2019


\bibitem{betsumiya2003} K. Betsumiya et al. ``On self-dual codes over some prime fields," {\it Discrete Math.,} 262(1-3), 37--58, 2003.


\bibitem{Conway1999} J. H. Conway, N. J. A. Sloane, {\it Sphere Packings, Lattices and Groups}, third ed., Springer, New York, 1999. 

\bibitem{Magma} J. Cannon, C. Playoust, {\it An Introduction to Magma.} University of Sydney, Sydney, Australia, 1994.

\bibitem{deboer1996} M.A. De Boer, ``Almost MDS codes," {\it Des. Codes Cryptogr.,} 9.2, 143-155, 1996.

\bibitem{Dodunekov2000} S. M. Dodunekov and N. L. Ivan, "Near-MDS codes over some small fields," {\it Discrete Math.,} 213.1-3, 55-65, 2000.

\bibitem{Gaborit2003} P. Gaborit and O. Ayoub, ``Experimental constructions of self-dual codes," {\it Finite Fields Appl.,} 9.3, 372-394, 2003.

\bibitem{Gaborit2002} P. Gaborit, ``Quadratic double circulant codes over fields," {\it J. Comb. Theory Ser. A}, 97(1), 85–107, January 2002.

\bibitem{Georgiou} S. Georgiou, C. Koukouvinos, ``MDS self-dual codes over large prime fields." {\it Finite Fields Appl.,} 8, 455–470, 2002.

\bibitem{grassl2008} M. Grassl. ``On self-dual MDS codes," {\it , in: ISIT 2008, Toronto, Canada, July 6–11, 2008,} pp. 1954-1957.

\bibitem{grassl2009} M. Grassl and T. A. Gulliver, ``On circulant self-dual codes over small fields," {\it Des. Codes Cryptogr.,} 52, 57, 2009. https://doi.org/10.1007/s10623-009-9267-1

\bibitem{Guenda2012} K. Guenda, ``New MDS self-dual codes over finite fields," {\it Des. Codes Cryptogr.,}  62, 31–42, 2012. https://doi.org/10.1007/s10623-011-9489-x

\bibitem{Gulliver2008} T.A. Gulliver, J.-L. Kim, and Y. Lee, ``New MDS or near-MDS self-dual codes." {\it IEEE Trans. Inform. Theory}, 54(9), 4354-4360, 2008.

\bibitem{gulliver2010} T. A. Gulliver and M. Harada, ``MDS self-dual codes of lengths 16 and 18," {\it IJICoT} 1.2, 208-213, 2010.

\bibitem{Han2008} S. Han and J.-L. Kim. ``On self-dual codes over ${\mathbb {F}} _5$," {\it Des. Codes Cryptogr.,} 48.1, 43-58, 2008.

\bibitem{haradaWeb} M. Harada and A. Munemasa, Database of self-dual codes, http://www.math.is.tohoku.ac.jp/~munemasa/selfdualcodes.htm

\bibitem{harada2003} M. Harada and P. R. Ostergard,  ``On the Classification of Self-Dual Codes over $\F_5$." {\it Graphs and Combinatorics,} 19(2), 203-214, 2003.

\bibitem{HP3}  W. C. Huffman and V. Pless. {\it Fundamentals of error-correcting codes}, Cambridge university press, 2010.

\bibitem{Huang2019} L. Huang, H. Zhang, R. Li, Y. Ge, and J. Wang, ``AI coding: Learning to construct error correction codes,'' {\it IEEE Transactions on Communications}, 68(1), 26-39, 2019


\bibitem{huffman2005} Huffman, W. C.  ``On the classification and enumeration of self-dual codes." {\it Finite Fields Appl.} {\bf 2005} {\it 11(3)}, 451--490.

\bibitem{Jin2017} L. Jin and C. Xing, ``New MDS self-dual codes from generalized Reed-Solomon codes," {\it IEEE Trans. Inform. Theory.} 63(3), 1434–1438, 2017

\bibitem{Kim2004} J.-L. Kim and Y. Lee. ``Euclidean and Hermitian self-dual MDS codes over large finite fields," {\it J. Combin. Theory Ser. A}, 105.1, 79-95, 2004. 


\bibitem{leon1982} J. S. Leon, V. Pless and N. J. A. Sloane, ``Self-dual codes over $GF(5)$," {\it J. Combin. Theory Ser. A}, 32(2), 178--194, 1982.



\bibitem{MacWilliam1972} F.J. MacWilliams, N.J.A. Sloane, J.G. Thompson, ``Good self-dual codes exist," {\it Discrete Math.,} 3, 153–162, 1972.


\bibitem{Melchor2011} A. Aguilar-Melchor, et al., ``Classification of Extremal and $s$-Extremal Binary Self-Dual Codes of Length 38," {\it IEEE Trans. Inform. Theory.} 58(4), 2253-2262, 2012

\bibitem{Mills2019} P. Mills, "Solving for multi-class using orthogonal coding matrices." SN Applied Sciences 1.11 (2019): 1451.

\bibitem{Nachmani2018} E. Nachmani, et al., ``Deep learning methods for improved decoding of linear codes,'' {\it IEEE Journal of Selected Topics in Signal Processing}, 12(1), 119-131, 2018.


\bibitem{Nebe2006} G. Nebe, E. M. Rains, and N. J. A. Sloane, {\it Self-dual codes and invariant theory,} volume 17, Berlin: Springer, 2006


\bibitem{park2011classification} Y. H. Park, ``The classification of self-dual modular codes." {\it Finite Fields Appl.,} 17(5), 442--460, 2011.


\bibitem{pless1975} V. Pless and N. J. A. Sloane. ``On the classification and enumeration of self-dual code," {\it J. Combin. Theory Ser. A}, 18.3, 313-335, 1975.


\bibitem{pless1987} V. Pless and V. Tonchev, ``Self-dual codes over $GF(7)$," {\it IEEE Trans. Inform. Theory}, 33(5), 723--727. 1987.

\bibitem{Shi2018} M. Shi, et al., ``Self-dual codes and orthogonal matrices over large finite fields", {\it Finite Fields Appl.,} 54, 297-314, 2018.

\bibitem{sok2019} L. Sok. ``Explicit constructions of MDS self-dual codes", {\it IEEE Trans. Inform. Theory}, 2019. DOI: 10.1109/TIT.2019.2954877

\bibitem{sok2020} L. Sok. ``New families of self-dual codes", In arXiv, 2005.00726.

\bibitem{Yan2019} H. Yan, ``A note on the constructions of MDS self-dual codes," {\it Cryptogr. Commun.,} 11, 259–268, 2019. https://doi.org/10.1007/s12095-018-0288-3


\end{thebibliography}





%
%

\end{document}